\documentclass[
    reprint,
    amsmath,
    amssymb,
    aps,
    prx,
    longbibliography,
    nofootinbib,
    notitlepage,
    superscriptaddress,
]{revtex4-2}

\usepackage{physics}

\usepackage[dvipsnames]{xcolor}
\definecolor{CSblue}{RGB}{13, 113, 242}

\usepackage{graphicx}
\graphicspath{ {img/} }
\usepackage{dcolumn}
\usepackage{bm}
\usepackage[colorlinks=true,urlcolor=blue,linkcolor=blue,citecolor=magenta]{hyperref}
\usepackage[normalem]{ulem}
\usepackage{tikz}
\usepackage{amsmath, amssymb, amsthm, mathtools}
\usepackage[T1]{fontenc}
\usepackage{bbold}
\usepackage{bm}
\usepackage{standalone}

\usepackage{algorithm2e}

\newcommand{\rz}{\right]}
\newcommand{\lz}{\left[}
\newcommand{\rd}{\right\}}
\newcommand{\ld}{\left\{}
\def\({\left(}
\def\){\right)}

\newcommand{\avg}[1]{\langle#1\rangle}

\newcommand{\customref}[2]{\hyperref[#1]{\ref*{#1}#2}}
\newcommand{\nn}{\nonumber\\}

\definecolor{Ured}{HTML}{cc0000}
\definecolor{Ublue}{HTML}{1f65cf}
\definecolor{Ugreen}{HTML}{198a11}
\definecolor{appBlue}{RGB}{68, 114, 196}
\definecolor{appOrange}{RGB}{214, 116, 50}

\definecolor{OliveGreen}{cmyk}{0.84, 0, 0.95, 0.30}

\renewcommand\vec{\bm}

\newcommand{\tmix}{\ensuremath{t_{\rm mix}}}

\newcommand{\bfr}{\mathbf{r}}

\newcommand{\decoder}{\Delta}
\newcommand{\contour}{\gamma}
\newcommand{\contourset}{\Gamma}
\newcommand{\state}{\vec \sigma}

\usepackage[capitalize,nameinlink]{cleveref}

\newtheorem{theorem}{Theorem}
\newtheorem{lemma}[theorem]{Lemma}
\newtheorem{proposition}[theorem]{Proposition}
\newtheorem{result}[theorem]{Result}

\newtheorem{apptheorem}{Theorem}[section]
\newtheorem{appcorollary}{Corollary}[section]

\def\figureautorefname~#1\null{Fig.~#1\null}
\def\equationautorefname~#1\null{Eq.~(#1)\null}
\def\sectionautorefname~#1\null{Sec.~#1\null}
\newcommand{\appref}[1]{\hyperref[#1]{App.~\ref*{#1}}}

\usepackage{enumerate}
\usepackage{xeCJK}

\begin{document}

\title{Slow mixing and emergent one-form symmetries in three-dimensional $\mathbb{Z}_2$ gauge theory}

\author{Charles Stahl}
\affiliation{Department of Physics, Stanford University, Stanford, CA 94305}
\author{Benedikt Placke}
\affiliation{Rudolf Peierls Centre for Theoretical Physics, University of Oxford, Oxford OX1 3PU, UK}
\author{Vedika Khemani}
\affiliation{Department of Physics, Stanford University, Stanford, CA 94305}
\author{Yaodong Li}
\affiliation{Department of Physics, Stanford University, Stanford, CA 94305}

\date{January 7, 2026}

\begin{abstract}

Symmetry-breaking order at low temperatures is often accompanied by slow relaxation dynamics, due to diverging free-energy barriers arising from interfaces between different ordered states.
Here, we extend this correspondence to classical topological order, where the ordered states are locally indistinguishable, so there is no notion of interfaces between them.
We study the relaxation dynamics of
the three-dimensional (3D) classical $\mathbb{Z}_2$ lattice gauge theory (LGT) as a canonical example.
We prove a lower bound on the mixing time in the deconfined phase, $t_{\text{mix}} = \exp [\Omega(L)]$, where $L$ is the linear system size.
This bound applies even in the presence of perturbations that explicitly break the one-form symmetry between different long-lived states. This perturbation destroys the \textit{energy barriers} between ordered states, but we show that entropic effects nevertheless lead to diverging \textit{free-energy barriers} at nonzero temperature.
Our proof establishes the LGT as a robust finite-temperature \textit{classical memory}. 
We further prove that entropic effects lead to an \emph{emergent} one-form symmetry, via a notion that we make precise.
We argue that the exponential mixing time follows from \textit{universal} properties of the deconfined phase,
and numerically corroborate this expectation by exploring mixing time scales at the Higgs and confinement transitions out of the deconfined phase. These transitions are found to exhibit markedly different \textit{dynamic} scaling, even though both have the \textit{static} critical exponents of the 3D Ising model.
We expect this novel entropic mechanism for memory and emergent symmetry to also bring insight into self-correcting quantum memories.

\end{abstract}

\maketitle

\section{Introduction}\label{sec:intro}

Many-body systems at low temperatures often exhibit ordering of one sort or another.
When standard Landau theory applies, 
the system is said to \textit{condense} collectively into one of many possible minima of free energy 
with each minimum carrying a distinct value of a \textit{local order parameter},
and we say the system exhibits \textit{spontaneous symmetry breaking} (SSB).
With the thermodynamic stability of the ordered states
often comes 
slow relaxational dynamics.
To go from one ordered state to another, the system usually has to go through stages of coexistence, but interfaces between different states come at a large (free) energy cost,
and hence take a long time to form. The timescale on which the system equilibrates to the equilibrium Gibbs ensemble --- the so-called \emph{mixing time}~\cite{peres2017markov} --- is therefore often a sensitive dynamical diagnostic of ordering (or the lack thereof). 

In this work, we ask if and how the above picture for thermodynamic order and slow relaxation extends beyond the standard Landau paradigm. In particular, much effort in condensed matter physics over the last several decades has been devoted to the study of \emph{topologically ordered} systems which lack a local order parameter~\cite{Wen2007book}. 
The focus has largely been on studying the equilibrium properties of these phases, and it has been shown that ground state (zero temperature) topological order is absolutely stable to small arbitrary perturbations. In contrast, until recently~\cite{ma2025circuitbasedchatacterizationfinitetemperaturequantum}, a 
general understanding of their dynamics is largely unknown except at unperturbed fixed points~\cite{DKLP2001topologicalQmemory, Alicki2010, Castelnovo2008, Temme2016}, and 
the dynamical stability to perturbations is almost entirely unexplored. Indeed, because different ordered states are \emph{locally indistinguishable}, it does not even make sense to talk about free-energy costs corresponding to interfaces between states. The study of slow dynamics and memory in topologically ordered systems is also essential for better understanding passive quantum error correction~\cite{brown2016RMP}, which is expected to require a notion of absolutely stable topological quantum order at finite temperature.

\begin{figure}[t]
    \centering
    \includegraphics{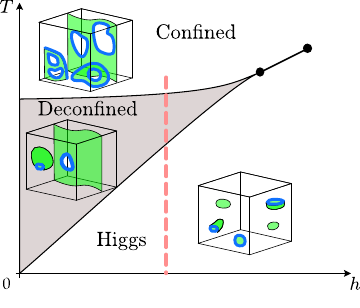}
    \caption{The ``re-entrant'' phase diagram of the 3D $\mathbb{Z}_2$ LGT in the presence of a symmetry breaking field, see \cref{sec:background} for details.
    Along the dashed line with a small but nonzero $h$, the model is in a trivial paramagnetic phase (called the Higgs phase) at $T = 0$, but at some $T > 0$ the deconfined phase reappears.
    We say the one-form symmetry is only emergent in the presence of thermal fluctuations, and this consequently leads to a robust, ``entropic'' classical memory.}
    \label{fig:phase-diagram-T-h}
\end{figure}

In particular, we study the dynamics of Wegner's $\mathbb{Z}_2$ lattice gauge theory (LGT)~\cite{wegner1971, OSTERWALDER1978440, marra1979statistical, fradkinsusskind1978, fradkinshenker1979, kogut1979rmp, trebst2007, tupitsyn2010, vidal2009phase, vidal2011phase} as a canonical example,
which has a \textit{classical topological ordered} (deconfined) phase. 
In modern parlance, the deconfined phase is understood as exhibiting spontaneous breaking of a so-called one-form\footnote{In $d$ dimensions, a $p$-form symmetry acts on deformable $(d-p)$-dimensional submanifolds of space. We refer the reader to \cref{sec:background} for a detailed review of basic notions of this model and its one-form symmetry.} symmetry~\cite{Nussinov2009, Gaiotto2015, lake2018higherform, mcgreevyannrev}.
A distinctive feature of such deformable ``higher-form'' symmetries is that they can \emph{emerge} at large scales~\cite{FRADKIN1990, huse2003coulomb, hermele2004coulomb} even if 
they are explicitly broken at microscopic scales, which is believed to be key to the absolute stability of these phases~\cite{Bravyi_2010, Bravyi_2011, Michalakis_2013}. This is well-illustrated by the phase diagram in \cref{fig:phase-diagram-T-h}, which shows that the deconfined phase survives even in the presence of a perturbation $h$  (the so-called Higgs coupling) which explicitly breaks the one-form symmetry. 
Likewise, the transitions corresponding to higher-form SSB  (and their associated critical properties) can also survive (sufficiently weak) explicit symmetry-breaking perturbations. This is in sharp contrast to conventional Landau SSB of global symmetries in which the ordered phase, and the slow dynamics that comes with it, is generically unstable to arbitrarily small symmetry-breaking perturbations.\footnote{This is encapsulated by the Gibbs phase rule~\cite{gibbs1878equilibrium}.}

The dynamical consequences of a symmetry-breaking field are much less clear. A nonzero $h$ eliminates any apparent \textit{energy barrier}
that grows with system size, which could potentially prevent rapid mixing between different long-lived states. 
Therefore, neglecting the effects of thermal fluctuations (which oftentimes accelerates equilibration), one might naively expect fast relaxation whenever $h > 0$, at all temperatures. Robust memory in this setting would, therefore, necessarily require a novel \emph{entropic} mechanism for slow mixing. Indeed, the ``re-entrant'' shape of the phase diagram in \cref{fig:phase-diagram-T-h} shows that thermal fluctuations are crucial to the deconfined phase, which only emerges at nonzero temperature for $h>0$. This equilibrium phase diagram  is hence an example of ``order-by-disorder''~\cite{chalker2017spin}, and we flesh out the dynamical consequences of this in detail in what follows.

Previous work by Poulin, Melko, and Hastings~\cite{poulin-melko-hastings}  found suggestive numerical evidence of slow mixing within the deconfined phase,
indicating that emergent symmetries, when spontaneously broken, also lead to slow relaxational dynamics, just like in standard Landau theory.
Nevertheless, 
detailed understanding is lacking and consensus is yet to be reached. This is particularly true because it is challenging to extract asymptotic behavior from numerical tests of slow dynamics which are sensitive to the existence of long-lived metastable states. A cautionary tale is the Ising model in a longitudinal field,
in which case the unfavored ground state relaxes to the ``true'' ground state in a time that is independent of system size, but exponentially large in the inverse external field strength. Similar concerns were also raised in Ref.~\cite{poulin-melko-hastings} regarding the $\mathbb{Z}_2$ LGT. In particular, their numerical results for the accessible system sizes and times suggest slow dynamics even beyond the deconfined phase, in the Higgs regime of \cref{fig:phase-diagram-T-h}.

In order to avoid the pitfalls of numerics, we focus on rigorous statements about the dynamics of the model.
Concretely, we study the Glauber dynamics of the $\mathbb{Z}_2$ LGT in 3D, which is a convenient mathematical model for 
time-dependent Ginzburg-Landau equations (TDGL) of ``Model A'' type~\cite{hohenberghalperin1977}.
They describe relaxational processes with non-conserved order parameters subject to local dissipation and/or thermal noise.
Within the mathematical literature, Glauber dynamics is an example of a reversible (i.e. obeying detailed balance) Markov chain, and has been studied extensively for Ising and Potts models, see e.g. Refs.~\cite{martinelli1999saintflour, peres2017markov} and references therein.
As such, they provide a natural setting for studying in detail the problem at hand.
On the other hand, mixing times of the Glauber dynamics for SSB of generalized symmetries, including higher-form symmetries, is an under-explored subject:
we do not know of field-theoretic descriptions of TDGL-type for its relaxation, nor rigorous bounds from the mathematical literature.

Our results may be summarized as follows: 

(1) We provide a proof of slow mixing in the 3D $\mathbb{Z}_2$ LGT, see \cref{thm:Ising_LGT_exp_tmix_x_general} for a formal statement.
In particular, we prove that a \textit{free energy barrier} can nevertheless be present within the deconfined phase, despite the lack of an \textit{energy barrier}
Consequently, the mixing time of its Glauber dynamics diverges exponentially with the system size.
To the best of our knowledge, such an \textit{entropic} mechanism of slow mixing has not been established before.
We expect this mechanism for protecting information might also bring valuable insights into self-correcting \textit{quantum} memories.

Our results thus rigorously establish the $\mathbb{Z}_2$ LGT as the first example of a passive classical memory, robust to symmetry breaking perturbations.

(2) We provide a proof of an emergent one-form symmetry --- induced by thermal fluctuations --- that transforms between topologically-distinct long-lived states, see \cref{lemma:bottle_symmetry} for a formal statement.
As we will show, our notion of emergent one-form symmetry is directly tied to time scales, which can in principle be measured either numerically in Monte-Carlo simulations, or even in experimental gauge-Higgs systems of amphiphilic films~\cite{huseleibler1991, huse1988phase}.
Our result thus provides a concrete dynamical diagnostic of  emergent symmetry.
This is particularly useful as the Wilson loop is no longer useful for diagnosing spontaneous breaking of one-form symmetries when it is explicitly broken by the Higgs term.
While several other diagnostics have been proposed~\cite{FredenhagenMarcu, gregor2011diagnosing,nahum2020selfdual,nahum2024patching}, their relations to directly measurable quantities are not clear to us.

(3) We supplement our proof with large-scale numerics near 
phase transitions out
of the deconfined phase, see \cref{result:different_tmix_scalings}.
In particular, we study both the transition as a function of temperature -- called the ``confinement'' transition -- and the transition as function of the external symmetry-breaking field -- called the ``Higgs'' transition.
We find that the memory times diverge as the transitions are approached from outside the deconfined phase, so that the dynamic phase diagram appear to agree with the thermodynamic one.
Although the Higgs and confinement transitions are both dual to the 3D Ising transition and have the same static critical exponents, their critical mixing times obey markedly different scaling laws.
We observe that while in both transitions the mixing time $\tmix$ is set by the correlation length $\xi$, the dependence is polynomial at the confinement transition ($\tmix = {\rm poly}(\xi)$) but exponential at the Higgs transition ($\tmix = {\rm exp}(\xi)$).
Thus, although we have argued for slow mixing within ordered phases and fast mixing otherwise, no such simple statements can be readily made for critical points, where relaxation dynamics contains information that cannot be completely determined by the static critical properties.
More generally, this distinction shows that different critical dynamics can emerge depending on whether an emergent symmetry is present at criticality.

The rest of this paper is organized as follows.

In \cref{sec:background}, we review the 3D $\mathbb{Z}_2$ LGT, its thermodynamic phase diagram, and the membrane representation of the partition function.
We also review the bottleneck theorem~\cite{Cheeger1971, LawlerSokal1988, JerrumSinclair1989} in \cref{sec:bottleneck} for lower bounding mixing times,  which we use extensively.
The theorem provides a general relation between mixing time scales and the stationary Gibbs distribution, 
and it makes precise intuitions provided by the Arrhenius law.

With these preparations, in \cref{sec:summary_of_results} we provide a formal statement and brief discussions of our results.
By identifying appropriate low-free-energy sectors (``bottles'') separated by high-free-energy barriers (``bottlenecks''),
we prove that within an open subregion of the deconfined phase
the mixing time $t_{\text{mix}}$ of single-spin-flip Glauber dynamics obeys a lower bound of the form 
\begin{align}\label{eq:mixing_time_intro}
    t_{\text{mix}} = e^{\Omega(L)},
\end{align}
where $L$ is the linear size of the system.

In \cref{sec:bottleneck_examples_Peierls}, we illustrate the use the bottleneck theorem for proving slow mixing for the 2D Ising model and for the 3D $\mathbb{Z}_2$ LGT with $h=0$.
Our proof uses tools that are motivated by universal properties of the deconfined phase, and this approach differs from existing ones in the literature.
For both models, we adapt the Peierls argument to establish the existence of bottlenecks between symmetry sectors.
Given the generality of the Peierls argument, this approach establishes \cref{eq:mixing_time_intro} for both models at the same time.

However, this approach does not straightforwardly apply in the presence of an explicit symmetry-breaking perturbation.
In \cref{sec:main_proof}, we show that the relevant symmetry is \emph{effectively} restored in the thermodynamic limit, and in doing so we make precise a notion of emergent one-form symmetry. 
To this end, we use Kramers-Wannier dualities of the 3D $\mathbb{Z}_2$ LGT to the three-dimensional Ising model, and a bound on non-local observables in the latter using cluster expansions.
The radius of convergence of this expansion sets the limits of validity for our rigorous results, which therefore cover only an open subregion of the deconfined phase. Despite this technical caveat, we expect slow mixing and our notion of emergent one-form symmetry
to persist throughout the entire phase, given that the physical inputs to the proof are universal.
We also expect that this line of reasoning --- emergent-symmetry-based Peierls arguments --- can be applied to a broader class of models.

In \cref{sec:critical_mixing_times}, we detail our numerical study of the confinement and Higgs transitions.

In \cref{sec:discussion}, we discuss open questions and possible future directions.

\tableofcontents

\section{Background\label{sec:background}}

\subsection{Three-dimensional Ising gauge theory as a membrane ensemble \label{sec:z2_LGT_intro}}

\begin{figure}[t]
    \centering
    \includegraphics{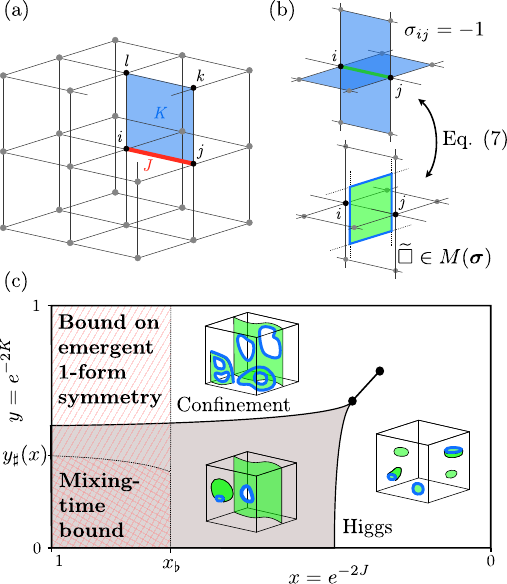}
    \caption{(a) Geometry of the 3D $\mathbb{Z}_2$ lattice gauge theory. (b) Mapping from spins to membrane representation. (c) Phase diagram of the model [\cref{eq:Z_def_gauge_fix}]. We also show a typical configuration for each region of the phase diagram, and indicate ranges of validity of our rigorous results. 
    The illustrations show typical configurations within each phase, as determined by $x$ (the fugacity of membranes) and $y$ (the fugacity of membrane boundaries).
    That is, the membrane area is suppressed by a small $x$, and the length of membrane boundaries is suppressed by a small $y$.
    \label{fig:lattice}}
\end{figure}

We consider the $\mathbb{Z}_2$ LGT defined on the 3D cubic lattice. Its geometry is illustrated in \cref{fig:lattice}(a).
The model includes two types of Ising spins, one type placed on each vertex (which we identify as ``matter'', denoted $\tau_i$), and one type placed on each edge (which we identify as ``gauge field'', denoted $\state_{ij}$). 
Throughout the paper we take periodic boundary condition in all three directions, and we take the linear size of the lattice to be $L$ in all three directions.
We may denote the lattice as $\Lambda = (\mathbb{Z}_L)^3$, where $\mathbb{Z}_L$ is the cyclic group of $L$ elements.
The partition function for this model is
\begin{align} \label{eq:Z_def}
    Z = \sum_{\state, \tau} \exp(K \sum_{\Box_{ijkl}} \state_{ij} \state_{jk} \state_{kl} \state_{li} + J \sum_{\avg{ij}} \state_{ij} \tau_i \tau_j).
\end{align}
Here, $K$ is the stiffness of gauge fluctuations, and $J$ is the strength of gauge-matter coupling (or Higgs coupling).
The model has a manifest gauge invariance,
\begin{align}
    & \tau_i \to \tau_i \chi_i,  \\
    & \state_{ij} \to \state_{ij} \chi_i \chi_j,
\end{align}
where $\chi_i = \pm 1$ can be chosen independently on each site.
This allows us to fix the gauge so that $\tau_i = +1$ for all $i$.
In this gauge, the Higgs coupling $J$ takes the form of a longitudinal magnetic field, acting (only) on the gauge field spins $\state$,
\begin{align} \label{eq:Z_def_gauge_fix}
    Z = \sum_{\state, \tau} \exp(K \sum_{\Box_{ijkl}} \state_{ij} \state_{jk} \state_{kl} \state_{li} + J \sum_{\avg{ij}} \state_{ij}).
\end{align}
We will henceforth work in the gauge $\tau_i = +1$ for all $i$.
Once we have picked a gauge, there is no gauge invariance left.
We can equivalently write the Hamiltonian for the model as 
\begin{equation}
H(\state) = -\sum_{\Box_{ijkl}} \state_{ij} \state_{jk} \state_{kl} \state_{li} - h \sum_{\avg{ij}} \state_{ij}, \label{eqn:LGThamiltonian}
\end{equation}
with $h=J/K$ and $T=1/K$, using the parameters in \cref{fig:phase-diagram-T-h}.

We adopt a convenient representation of the partition function in terms of an ensemble of membranes.
To each bond $\avg{ij}$ we associate its dual plaquette on the dual lattice, and and to each plaquette we associate a dual bond, see \cref{fig:lattice}(b).
For each spin configuration $\state$, we assign a subset of dual plaquettes (they together constitute a membrane, possibly with a boundary) as follows,
\begin{align}
    M(\state) & = \{ \widetilde{\Box} : \widetilde{\Box} \text{ is dual to the primal} \nn
    & \hspace{.5in} \text{bond $\avg{ij}$ and } \state_{ij} = -1\}.
\end{align}
That is, we say a dual plaquette $\widetilde{\Box}$ is occupied by the membrane $M$ if and only if the corresponding gauge field spin is $\sigma_{ij} = -1$.
Such membrane configurations $M$ are in 1-to-1 correspondence with spin configurations $\state$.
Similarly, we define $\partial M$ as the boundary of $M$, which is a subset of edges of the dual cubic lattice, forming (possibly disjoint) closed loops (i.e. ``flux loops''),
\begin{align}
    \partial M(\state) = & \{ \avg{\widetilde{ij}} : \avg{\widetilde{ij}} \text{ is dual to the primal plaquette} \nn 
    & \hspace{.5in} \text{$\Box_{ijkl}$ and } \state_{ij} \state_{jk} \state_{jl} \state_{li} = -1\}.
\end{align}
An example membrane configuration for a single flipped spin $\sigma_{ij} = -1$ is indicated in \cref{fig:lattice}(b) in green, with the boundary indicated in blue.

With these definitions, the partition function in \cref{eq:Z_def} can be rewritten in terms of membrane configurations~\cite{huseleibler1991, nahum2020selfdual}
\begin{align} \label{eq:Z_mem_expansion}
    Z \propto \sum_{M} x^{|M|} \cdot y^{|\partial M|}.
\end{align}
where, matching the Boltzmann weights with \cref{eq:Z_def} we have 
\begin{equation}
x = e^{-2 J}, \quad y = e^{-2 K}.
\end{equation}
\cref{eq:Z_mem_expansion} can be viewed as a ``low temperature'' expansion of \cref{eq:Z_def}. 
The parameters $x$ and $y$ are the \emph{fugacities} of membranes and membrane boundaries, respectively, which control how ``tight'' or ``floppy'' the geometrical objects are. When fugacity is low, typical configurations contain large objects, while when fugacity is large, large objects are penalized.

See \cref{fig:lattice}(c) for the phase diagram of the model, together with illustrations of typical membrane configurations within each phase.
This phase diagram is related to that in \cref{fig:phase-diagram-T-h} (in the temperature-field plane) by identifying $K = 1/T$ and $J = h/T$. 
The deconfined phase is therefore an ensemble of large membranes with very small punctures in them. The membranes are tension-free and are easily able to span across the system.
The line tension of membrane boundaries (i.e. flux lines) vanishes at the confinement transition, so that in the confined phase there are large membranes with large punctures.
At the Higgs transition, the membrane surface tension becomes nonzero, and we have only small membranes with small punctures in the Higgs phase.
These two separate transitions out of the deconfined phase join at a multi-critical point. 
However, there are no thermodynamic phase transitions separating the Higgs and confined phases far away from the deconfined phase~\cite{OSTERWALDER1978440, fradkinshenker1979}.
While the typical membrane configurations deep in the Higgs and confined phases look very different, they are smoothly connected by an intermediate regime featuring configurations of long membrane strips (the ``seaweed'' of Ref.~\cite{gregor2011diagnosing}).

\begin{figure}[h]
    \centering
    \includegraphics{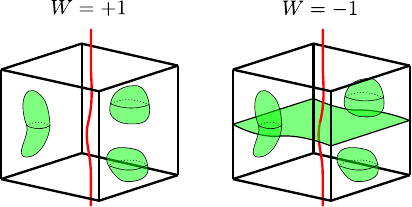}
    \caption{Example of the two homologically inequivalent classes of ground states in our model.
    All membranes illustrated here are boundary-less ($\partial M = \emptyset$). 
    The two classes are distinguished by a homologically nontrivial Wilson loop along the $z$ direction, which takes values $+1$ and $-1$ respectively.
    }
    \label{fig:logical_states}
\end{figure}

When the external field is zero ($h=0$ or equivalently $x = 1$), the energy is exactly invariant under the \emph{symmetry transformation} $\state_{ij} \to \state_{ij}\chi_i\chi_j$.\footnote{This looks like the gauge transformation, but recall that we have already fixed the gauge so that there is no gauge invariance left.}
In the membrane representation, this transformation flips a collection of contractible boundary-less membranes, where a contractible membrane is one that can be smoothly contracted to a point, such as those on the left side of \cref{fig:logical_states}. Contractible membranes are also called homologically trivial. The energy is also invariant under the transformation of flipping an entire noncontractible (or homologically nontrivial) boundary-less
membrane, such as on the right of \cref{fig:logical_states}.
These operations are generators of the so-called one-form symmetries~\cite{mcgreevyannrev}.

The one-form symmetry also comes with a \emph{charged operator}, a Wilson loop along a noncontractible loop. This operator is charged under the one-form symmetry in the sense that it changes sign upon flipping a noncontractible membrane. Furthermore, when all membranes are boundary-less (at $T=0, K=\infty$), the Wilson loop $W$ partitions states into two symmetry sectors $W=+1$ and $W=-1$. However, when boundaries are present, the distinction between sectors and the definition of the symmetry charge become less clear. 

Once $h>0$, the energy depends on the total membrane area, so the model is no longer exactly invariant under the one-form symmetry. We say that the field explicitly breaks the one-form symmetry.\footnote{To deal with these subtleties, we propose a notion of emergent one-form symmetry for the ensemble \cref{eq:Z_mem_expansion} of  tensionful membranes, possibly with boundaries, see the statement of \cref{lemma:bottle_symmetry}.}

While we will mostly focus on the representation in \cref{eq:Z_mem_expansion}, we mention that the model admits a different representation in terms of primary (rather than dual) membranes $\mathcal{M}$,
\begin{align} \label{eq:Z_mem_expansion_HTE}
    Z \propto \sum_{\mathcal{M}} (x')^{|\mathcal{M}|} (y')^{|\partial \mathcal{M}|},
\end{align}
where $x' = \tanh K$ and $y' = \tanh J$.
This can be obtained from a ``high-temperature'' expansion of \cref{eq:Z_def}.
Comparing \cref{eq:Z_mem_expansion} and \cref{eq:Z_mem_expansion_HTE}, we see that the model has a self-duality~\cite{nahum2020selfdual}.

\subsection{Single-spin-flip Glauber dynamics \label{sec:glauber-metropolis}}

We now define the Glauber dynamics of the Ising gauge theory with Metropolis update rules~\cite{enwiki:1320720012}.
This is an example of a Markov chain that is \textit{reversible} with respect to the Gibbs distribution $\mu(\vec\sigma) := Z^{-1} \exp[-\beta H(\vec\sigma)]$, with the energy $H(\state)$ defined in \cref{eqn:LGThamiltonian}; a formal discussion of reversible Markov chains will be given in \cref{sec:bottleneck}.

After initialization, the dynamics proceeds by repeatedly updating a single spin $\state_{ij}$ at a time.  
First, we chose a link $\avg{ij}$ uniformly at random.
Let $\Delta E$ be the energy change when flipping $\state_{ij} \mapsto -\state_{ij}$ while keeping all other spins fixed.
The \emph{Metropolis rule} is then to flip $\sigma_\ell$ with probability
\begin{align} \label{eq:metropolis}
\mathbb{P}(\state_{ij} \mapsto -\state_{ij}) = \min \bigl\{1, e^{- \Delta E}\bigr\}.
\end{align}
Iterating this procedure across the lattice causes the system to
equilibrate to the stationary distribution, at long times.
We adopt the convention that each unit time step contains $\Theta(L^3)$ attempts of spin flips.

The Metropolis rule defines an \textit{ergodic} Markov chain, i.e. it has a \textit{unique} stationary distribution~\cite{enwiki:1320720012}, 
namely the Gibbs distribution $\mu(\state)$.

We remark that there are several different rules other than Metropolis that can be used to define Glauber dynamics.
For example, the \textit{heat bath rule} resamples all spins within a finite-size box, while keeping all spins outside fixed.
While we focus on single-site Glauber (Metropolis) dynamics, we believe our results hold more generally for all local, detailed balance dynamics.

\subsection{Mixing time from bottleneck ratio for reversible Markov chains} \label{sec:bottleneck}

We collect a few basic formal notions and general results for reversible Markov chains.

Consider a Markov chain $(\Omega, P)$ with state space $\Omega$ and transition matrix $P$.
Its matrix elements are non-negative, where $P(\state, \state')$ denotes the transition probability from $\state$ to $\state'$ within one step.
We have, by conservation of probability,
\begin{align}
    \forall \state \in \Omega, \quad \sum_{\state'} P(\state, \state') = 1.
\end{align}

We focus on the case when $(\Omega, P)$ is ergodic so that there is a unique stationary distribution
$\mu$ such that
\begin{align}
    \mu P = \mu.
\end{align}
The Markov chain $(\Omega, P)$ is said to be \emph{reversible with respect to its stationary distribution $\mu$} if it satisfies the \emph{detailed balance condition}
\begin{align} \label{eq:detailed_balance}
    \forall \state, \state' \in \Omega, \quad \mu (\state) P(\state, \state') = \mu (\state') P(\state', \state).
\end{align}
Intuitively, the above condition means that the probability of observing a particular time series of states $\{\sigma_t\}_t$ is identical to the probability of observing the time-reversed series, hence the name.

We define the \emph{mixing time} of the pair $(\Omega, P)$ as
\begin{align}
    t_{\rm mix}(\varepsilon) \coloneqq \inf \{t: \max_{\state \in \Omega} \norm{P^t(\state, \cdot) - \mu}_{\rm TV} \leq \varepsilon\}.
\end{align}
Here $\norm{\cdot}_{\rm TV}$ denotes the total variation norm, and $\varepsilon$ is an abitrary cutoff conventionionally taken to be $1/4$.
The choice of the constant $\epsilon$ is unimportant, since the particular choice simply rescales the value by a constant independent of the chain $t_{\rm mix}(\epsilon) = \Theta(t_{\rm mix}(1/4) \log \epsilon^{-1})$ \cite[Section 4.5]{peres2017markov}.

The physical meaning of $t_{\rm mix}(\varepsilon)$ should also be clear: when the Markov chain is run for time $t > t_{\rm mix}(\varepsilon)$, it generates a distribution that is statistically indistinguishable (up to $\epsilon$) from the stationary one.

\begin{figure}[b]
    \centering
    \includegraphics{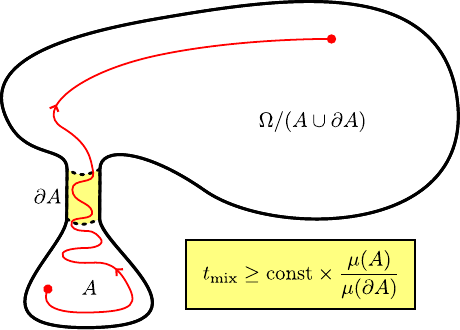}
    \caption{Illustration of the bottleneck theorem [\cref{thm:tmix_bottleneck}].
    }
    \label{fig:bottleneck}
\end{figure}

A particularly physical way to derive lower bounds on the mixing times is via \emph{bottlenecks}. This method is handy since it requires knowledge only about the structure of the steady-state distribution $\mu$.
For a subset of state space $A \subset \Omega$, we define the \emph{bottleneck ratio}
\begin{align}
    \label{eq:def_phi_A}
    & \Phi(A) \coloneqq \frac{\sum_{\state \in A, \state' \in \Omega \setminus A} \mu(\state) P(\state,\state')}{\mu(A)},
\end{align}
as the ratio of the probability flow out of $A$ to the probability weight in $A$.
The summation over $\state'$ may be restricted to $\state' \in \partial A$, defined as follows
\begin{align} \label{eq:def_bottleneck}
    \partial A \coloneqq \{\state' \in \Omega \setminus A: \exists \state \in A \text{ s.t. } P(\state,\state') > 0\}.
\end{align}
For a reversible Markov chain satisfying \cref{eq:detailed_balance}, we can rewrite \cref{eq:def_phi_A} as
\begin{align}
    \label{eq:phi_A_bottleneck}
    \Phi(A) =&\, \frac{\sum_{\state \in A} \sum_{\state' \in \partial A} \mu(\state) P(\state,\state')}{\mu(A)} \nn
    =&\, \frac{\sum_{\state' \in \partial A} \mu(\state') \sum_{\state \in A} P(\state',\state)}{\mu(A)} \nn
    \leq&\, \frac{\sum_{\state' \in \partial A} \mu(\state')}{\mu(A)}
    \equiv \frac{\mu(\partial A)}{\mu(A)}.
\end{align}
For intuitive reasons we will call $\partial A$ a \textit{bottleneck}, and $A$ a \textit{bottle}; see \cref{fig:bottleneck} for an illustration. In a colloquial sense, $\Phi(A)$ measures the probability flow out of $A$, and its inverse captures the lifetime of an initial state within $A$ to remain in $A$.
For the entire state space, we define the bottleneck ratio of $(\Omega, P)$ to be the minimum over all subsets that contain less than half of the total weight
\begin{align}
    \label{eq:def_phi_ast}
    & \Phi_\ast \coloneqq \min_{A \subset \Omega: \mu(A) \leq 1/2} \Phi(A),
\end{align}
which should capture the longest lifetime in the system.

It is intuitive that a small bottleneck ratio should imply slow mixing. This is because a small bottleneck ratio as defined in \cref{eq:def_phi_ast} implies that there are parts of state space which contain less than half the weight, but at the same time have high weight compared to the probability flow into (and by reversibility out of) them. 
Making this intuition precise,
we recall a basic result relating the bottleneck ratio and the mixing time~\cite{Cheeger1971, LawlerSokal1988, JerrumSinclair1989}
\begin{theorem}
\label{thm:tmix_bottleneck}
The mixing time is lower bounded by the inverse bottleneck ratio up to a constant factor,
\begin{align}
    \lz 4 \cdot t_{\rm mix}(1/4) \rz^{-1} \leq  \Phi_\ast \leq \frac{\mu(\partial A)}{\mu(A)}, \quad \forall A \subset \Omega, \mu(A) \leq 1/2.
\end{align}
\end{theorem}
\noindent A proof can be found in e.g.~\cite[Section 7.2]{peres2017markov}.

\section{Formal statement of main results \label{sec:summary_of_results}}

1. Our main result lower bounds the mixing time of the $\mathbb{Z}_2$ LGT  within the deconfined phase.
The theorem is stated as follows (we use the standard big-O notation for describing asymptotic behavior of functions)
\begin{theorem}[Slow mixing of $\mathbb{Z}_2$ LGT]
\label{thm:Ising_LGT_exp_tmix_x_general}
For the $\mathbb{Z}_2$ LGT on $\Lambda = (\mathbb{Z}_L)^3$, there exists $x_\flat \in (0, 1)$ such that for all $x \in (x_\flat, 1]$, there exists $y_\sharp(x) \in (0,1)$ 
for which 
\begin{align}
     & x \in (x_\flat, 1], y \in [0, y_\sharp(x)) 
     \quad  \Rightarrow \quad  t_{\rm mix} = \exp( \Omega(L) ).
\end{align}

\end{theorem}
\noindent 
\cref{thm:Ising_LGT_exp_tmix_x_general} applies to an open subregion of the deconfined phase, as we show in \cref{fig:lattice}(d).

~

2. An important input to proving \cref{thm:Ising_LGT_exp_tmix_x_general} is a notion of emergent one-form symmetry.
To make this precise, we consider the following \textit{conditional ensemble} for a fixed $\Gamma$ (i.e. a collection of flux loops), by summing over \textit{all} membranes $M$ with $\Gamma$ as its boundary,
\begin{align}\label{eq:Z_with_boundary}
    Z(x, \Gamma) \coloneqq \sum_{M : \partial M = \Gamma} x^{|M|}.
\end{align}
Since we have fixed $\Gamma$, it is unambiguous to talk about the homology classes of each membrane appearing in the conditional ensemble.
We therefore define $Z_\pm(x, \Gamma)$ to be the corresponding weights of the two homology classes, where $Z(x, \Gamma) = Z_+(x, \Gamma) + Z_-(x, \Gamma)$.\footnote{More precisely, given $M_1, M_2$ where $\partial M_1 = \partial M_2 = \Gamma$, it is unambiguous wether they belong to the same homology class or not, e.g. by examine the value of $W$ for the state $M_1 \oplus M_2$. Whether a given state is sorted into $Z_+(x, \Gamma)$ or $Z_-(x, \Gamma)$ depends on the particular decoder (a concrete decoder is detailed in \cref{sec:main_proof}), but we note that our result \cref{lemma:bottle_symmetry} bounds their ratio from both sides, and is uniform in $\Gamma$, therefore is independent of the decoder.}
We show that
\begin{lemma}[Emergent one-form symmetry]
\label{lemma:bottle_symmetry}
Let $x > x_\flat$, where $x_\flat$ is given in \cref{thm:Ising_LGT_exp_tmix_x_general}.
The conditional ensemble [\cref{eq:Z_with_boundary}] has asymptotically equal weight within each bottle.
More precisely, we have
\begin{align}
\label{eq:bottle_symmetry}
\forall \Gamma, \quad
    \abs{\frac{Z_+(x, \contourset)}{Z_-(x, \contourset)} - 1} = \exp(-\Omega(L)).
\end{align}
Note that this result depends only on $x$, but not on $y$; therefore, it holds for a region beyond the deconfined phase.
Note also that this result holds for all $\Gamma$.
\end{lemma}
\noindent 
Physically, it is precisely the entropic contributions of the ensemble of membranes (conditional on $\partial M = \Gamma$) in \cref{eq:Z_with_boundary} that cancels out the bare surface tension ($x < 1$), and \cref{lemma:bottle_symmetry} may be interpreted as a notion of renormalized surface tension being zero. 

We note that this notion depends only on $x$, but not $y$, and therefore extends out of the deconfined phase into the confined phase. Thus, while the deconfined phase is the SSB phase of the emergent symmetry, by our definition the symmetry is emergent even when it is not spontaneously broken.

Physically, we expect some notion of emergent symmetry in the entire deconfined phase (and a subregion of the confined phase); but \cref{eq:bottle_symmetry} is too strong a condition to ask for all $\Gamma$, regardless of their shape or length.
Some relaxations of \cref{eq:bottle_symmetry} will need to be made for this purpose, particularly in the vicinity of the multicritical point.

~

3. Our third result is on critical mixing times, obtained from numerical memory experiments.
\begin{result}[Critical mixing times]\label{result:different_tmix_scalings}
The mixing time near the confinement transition exhibits a critical slowing down
\begin{align}
    \tmix^{\rm C} = \Omega(\mathrm{poly}(\xi)).
\end{align}
where as near the Higgs transition we have an unconventional, exponential slow down,
\begin{align}
    \tmix^{\rm H} = \exp(\Omega(\xi)).
\end{align}
Here, $\xi$ is the correlation length.
We support our results with scaling arguments and data collapses.
\end{result}
\noindent Recall that the two transitions have the same static critical exponents, and are electro-magnetic duals to each other.
Our result is compatible with this since Glauber dynamics is not invariant under duality. It illustrates that critical relaxation dynamics generically contains nontrivial information about the dynamics which is not determined by the static critical properties.

\section{Warmup examples of bottles and bottlenecks \label{sec:bottleneck_examples_Peierls}}

We discuss examples of Glauber dynamics and lower bounds on their mixing times using \cref{thm:tmix_bottleneck}.
We also show how the 2D Ising model can be described using a membrane ensemble similar to the ensemble for 3D gauge theory, and how the Peierls argument leads to a symmetry-based
bottleneck proof of slow mixing.

\subsection{2D Ising model at large $\beta$}\label{sec:Ising}

The slow mixing of 2D Ising model for all $\beta > \beta_c$ has been studied extensively in the literature, see \cite[ Chapter 15]{peres2017markov} for a survey.
Detailed understanding of the model is often needed to access $\beta \to \beta_c$.
Here, we present an alternative proof based on the Peierls argument and the Ising symmetry.
Our proof only works for sufficiently large $\beta$ and is therefore weaker than existing results.
The purpose of this exercise is to introduce the necessary concepts, which are generalizable for the purpose of proving slow mixing in the gauge theory.

We describe the 2D Ising model as a model of membranes with punctures, like the gauge theory: the all-up state has no membranes and the all-down state has a single unbroken membrane. 
A typical state in the membrane representation is shown in \cref{fig:ising_config}.
It takes a long time to mix the two ordered state because, in going between the two states, a large domain wall (puncture) must be formed.
We will therefore define  \textit{bottlenecks} as states with \textit{long} domain walls, and \textit{bottles} as those with only \textit{short} domain walls.

\begin{figure}[b]
    \centering
    \includegraphics[]{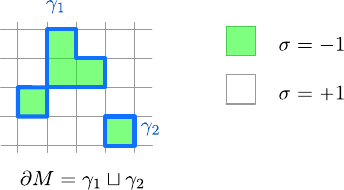}
    \caption{Geometry and notation of the two-dimensional Ising model.}
    \label{fig:ising_config}
\end{figure}

\begin{figure*}
    \centering
    \includegraphics[]{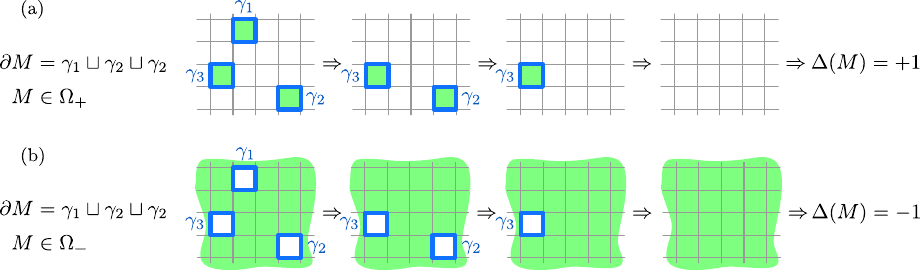}
    \caption{(a) Illustration of the decoder in \cref{eq:Ising_decoder}. (b) Same decoder but acting on state from a different bottle. }
    \label{fig:ising_decoder}
\end{figure*}

Consider a discretization of the 2-torus with lattice $\Lambda = \mathbb{Z}_L \times \mathbb{Z}_L$, where $\mathbb{Z}_L$ is the cyclic group of $L$ elements.
The state space is $\Omega = \{+1, -1\}^\Lambda$, and each $\state \in \Omega$ is a \textit{spin configuration}, with $\state_\bfr = \pm 1$ being the spin value on site $\bfr \in \Lambda$.

We denote the \textit{Hamming weight} of $\state$ as $|\state|$.
For $\state, \state' \in \Omega$, we define the binary operation $\oplus: \Omega \times \Omega \to \Omega$ such that
\begin{align}
    (\state \oplus \state')_\bfr = \state_\bfr \cdot \state'_\bfr.
\end{align}

To make the membrane analogy as explicit as possible, we interpret each site $\bfr$ as a plaquette. For each spin configuration $\state$, we say that a plaquette is occupied by a membrane if and only if $\state_\bfr=-1$.
With the membrane representation in mind, we will henceforth denote each $\state$ with a membrane $M$.
Then the partition function is
\begin{align} 
\label{eq:Ising_partition_function}
    Z &= \sum_{\state} \exp(K \sum_{\avg{ij}} \state_{i} \state_{j}  + J \sum_{i} \state_{i} ) %
    \propto \sum_{M} x^{|M|} \cdot y^{|\partial M|},
\end{align}
with constants defined similarly to before. The Ising model is only a classical memory for $J=0$ ($x=1$), so we focus on that case here.

For each $M \in \Omega$, we denote its corresponding \textit{domain wall} (DW) as $\partial M$; the notation suggests that the DW can also be thought of as a membrane boundary.
An illustration is provided in \cref{fig:ising_config}.
In general, $\partial M$ is a disjoint union of connected contours on the dual lattice,
\begin{align}
    \label{eq:dw_decomposition_connected}
    \partial M = \bigsqcup_{\alpha} \contour_\alpha.
\end{align}
Here, for a finite system size $L$, the index $\alpha$ takes values within a finite index set.
We will refer to a connected contour $\contour$ on the dual lattice as a \textit{connected DW}.\footnote{We do not try to ``resolve'' domain wall crossings: if four DWs share a vertex, they all belong to the same component.}
For each connected DW $\contour$,
we define
\begin{align}
    \delta(\contour) \coloneqq \mathrm{argmin}_{M \in \Omega: \partial M = \contour} |M|.
\end{align}
Whenever there are more than one $M$ with minimal weight (or none, if $\contour$ is not contractible), we say $\delta(\contour)$ is undefined.
We make the following observation
\begin{proposition} \label{prop:Ising_short_DW}
On $\Lambda = \mathbb{Z}_L \times \mathbb{Z}_L$, for each connected DW $\contour$,
\begin{align}
|\contour| \leq \frac{L}{10}
\  &\Rightarrow \  
\text{$\delta(\contour)$ is defined, and }
    |\delta(\contour)| \leq \frac{L^2}{1600},\\
|\contour| \leq \frac{L}{2}
\  &\Rightarrow \  
\text{$\delta(\contour)$ is defined, and }
    |\delta(\contour)| \leq \frac{L^2}{64}.
\end{align}
\end{proposition}
\noindent With this Proposition, we define $\ell_\ast \coloneqq L/10$ as the cutoff length for connected DWs.
We say a connected DW $\contour$ is \textit{short} if $|\contour| \leq \ell_\ast$, and \textit{long} if $|\contour| > \ell_\ast$. 
Long domain walls may either be contractible or noncontractible.

With these definitions we can start to construct the bottles and bottlenecks.
For $M \in \Omega$, 
we will use $\contour_\alpha \subset \partial M$ to refer to the connected components of $\partial M$, having in mind that 
$\partial M = \bigsqcup_{\alpha} \contour_\alpha$.
We define the \textit{long (connected) DW count} for $M \in \Omega$ as 
\begin{align}
    N_{\rm LDW}(M) \coloneqq |\{\alpha: \contour_\alpha \subset \partial M, |\contour_\alpha| > \ell_\ast \}|.
\end{align}
We define the subspace of \textit{short DW states},
\begin{align} \label{eq:Ising_bottles}
    \Omega_0 &\coloneqq 
    \{M \in \Omega : N_{\rm LDW}(M) = 0 \}.
\end{align}
For each short DW state $M \in \Omega_0$, we define its image under a decoder $\decoder: \Omega_0 \rightarrow \Omega$
denoted $\decoder M \in \Omega$, as
\begin{align} \label{eq:Ising_decoder}
    \decoder M \coloneqq M \oplus \left(\bigoplus_{\alpha} \delta (\contour_\alpha) \right), \quad \text{where } \partial M = \bigsqcup_{\alpha} \contour_\alpha.
\end{align}
Our decoder is of ``loop-by-loop'' type~\cite{nahum2024patching}, see illustrations in \cref{fig:ising_decoder}.
We have the following property of $\decoder M$,
\begin{proposition} \label{prop:Ising_decoder}
$\forall M \in \Omega_0$, $|\decoder M| = 0$ or $|\decoder M| = |\Lambda|$.
That is, $\decoder M$ is either the ``all up'' state, or the ``all down'' state.
\end{proposition}
\noindent Therefore, with a slight abuse of notation, we view $\decoder$ as a map $\decoder: \Omega_0 \to \{+1, -1\}$, which we call a \textit{decoder}.
With the decoder we define the \textit{bottles} as
\begin{align}
    \label{eq:Ising_plus_minus_bottles_def}
    \Omega_\pm \coloneqq \decoder^{-1}(\pm 1).
\end{align}
By Ising symmetry, we have that 
\begin{align}
    \label{eq:decoder_condition_2}
    \mu(\Omega_+) = \mu(\Omega_-) = (1/2) \cdot \mu(\Omega_0) \leq 1/2.
\end{align}
We also consider the \textit{subspace of states having exactly one long connected DW},
\begin{align} \label{eq:Ising_Omega_1}
    \Omega_1 \coloneqq
    \{M \in \Omega : N_{\rm LDW}(M) = 1 \}.
\end{align}

Our next Proposition states that the two bottles cannot be connected directly under single-spin-flip dynamics without first going through $\Omega_1$.
\begin{proposition}
\label{prop:Ising_bottleneck_condition}
The bottles $\Omega_\pm$ as defined in \cref{eq:Ising_plus_minus_bottles_def} satisfy the following condition under Glauber dynamics $(\Omega, P)$
\begin{align}
    \label{eq:decoder_condition_1}
    \forall M \in \Omega_+, \quad P(M,M') > 0 
    \ \Rightarrow\ 
    M' \in \Omega_+ \cup \Omega_1.
\end{align}
A similar results holds for $M \in \Omega_-$.
\end{proposition}
\begin{proof}
Let $M \in \Omega_+$, and suppose $P(M, M') > 0$ for some state $M'$.
We first note that $M' \in \Omega_+ \cup \Omega_- \cup \Omega_1 = \Omega_0 \cup \Omega_1$.
This is because a single spin flip can only leave $\Omega_0$ by (1) pushing a single loop over the limit, therefore entering $\Omega_1$, or (2) by combining multiple short loops into a \textit{single} long loop, therefore also entering $\Omega_1$.

Based on this observation, we examine $|M \oplus \decoder M|$ and $|M' \oplus \decoder M'|$, i.e. the change in Hamming weight of $M, M'$ under the decoder.
Recall that by \cref{eq:Ising_decoder} we have
\begin{align}
    \decoder M \oplus M = \bigoplus_{\alpha} \delta (\contour_\alpha), \\
    \decoder M' \oplus M' = \bigoplus_{\alpha} \delta (\contour'_\alpha).
\end{align}
In case (1), $\decoder M \oplus M$ and $\decoder M' \oplus M'$ differ on exactly one connected contour, which we call $\gamma_1$ and $\gamma_1'$, respectively.
We have $|\gamma_1| \leq \ell_\ast = L / 10$ in $M$, by definition; and after the spin flip we have $|\gamma_1'| \leq \ell_\ast + 4 \leq L / 2$ in $M'$.
Thus,
\begin{align}
    & |\decoder M \oplus M \oplus \decoder M' \oplus M'| \leq 
    |\delta (\gamma'_1)| + |\delta (\gamma_1)| \nn
    & \leq \frac{1}{64} \cdot L^2 + \frac{1}{1600}\cdot L^2 \leq \frac{1}{60} \cdot L^2.
\end{align}
In case (2), the combined long DW in $M'$ comes from at most 4 short DWs in $M$, and therefore $|\gamma_1'| \leq 4 \ell_\ast \leq L / 2$, and by \cref{prop:Ising_short_DW} we have $|\delta(\gamma)| \leq L^2/64$. Therefore, by the loop-by-loop nature of the decoder, we have
\begin{align}
    & |\decoder M \oplus M \oplus \decoder M' \oplus M'| \leq  |\delta (\gamma'_1)| + 
    \sum_{\alpha=1}^{4} |\delta (\gamma_\alpha)| \nn
    & \leq \frac{1}{64} \cdot L^2 + 4 \cdot \frac{1}{1600}\cdot L^2 \leq \frac{1}{50} \cdot L^2.
\end{align}

On the other hand, if $M' \in \Omega_-$
\begin{align}
|M \oplus \decoder M \oplus & M' \oplus \decoder M'| \nn
    =\, &
    |(M \oplus M') \oplus (\decoder M \oplus \decoder M')| \nn
    =\, &
    |\Lambda| - 1 = L^2 - 1,
\end{align}
which is a contradiction.
It follows that $M' \in \Omega_+ \cap \Omega_1$.
\end{proof}

Comparing \cref{eq:def_bottleneck}, we have by \cref{prop:Ising_bottleneck_condition} that
\begin{align}
    \label{eq:bottleneck_subset_Omega_1}
    \partial \Omega_\pm \subset \Omega_1 \quad \Rightarrow \quad \mu(\partial \Omega_\pm) \leq \mu(\Omega_1).
\end{align}
As $\mu(\Omega_+) = \mu(\Omega_-) \leq 1/2$ from \cref{eq:decoder_condition_2}, we have upon comparing 
\cref{eq:def_phi_ast,eq:def_phi_A,eq:def_bottleneck,eq:phi_A_bottleneck,eq:bottleneck_subset_Omega_1} that
\begin{align} \label{eq:Ising_bottleneck_ratio_bound}
    \Phi_\ast \leq \Phi(\Omega_\pm) \leq \frac{\mu(\partial \Omega_\pm)}{\mu(\Omega_\pm)} 
    \leq \frac{\mu(\Omega_1)}{\frac{1}{2} \mu(\Omega_0)}.
\end{align}
The last ratio is upper bounded by a Peierls' inequality, as follows.

\begin{proposition}
\label{thm:bottleneck_ratio_Ising_Peierls}
For sufficiently large $\beta$, there exists $c > 0$ and $L_0 \in \mathbb{Z}$, $L_0 < \infty$ such that
\begin{align}
    \forall L > L_0, \quad \frac{1}{2} \Phi_\ast \leq \frac{\mu(\Omega_1)}{\mu(\Omega_0)} \leq e^{-c L}.
\end{align}
\end{proposition}
\begin{proof}
We use an argument similar to Peierls' argument.
Observe that
\begin{align} \label{eq:Ising_Peierls_series}
    \frac{\mu(\Omega_1)}{\mu(\Omega_0)}
    & \leq \sum_{\text{contour $\contour$}} \mathbb{1}_{|\contour| \geq \ell_\ast} y^{|\contour|} \nn
    & \leq (L^2) \sum_{\ell = \ell_\ast}^{\infty} (z-1)^\ell y^\ell,
\end{align}
where $y = e^{-2K}$ is the fugacity of a unit domain wall, see \cref{eq:Ising_partition_function}, and $z = 4$ the coordination number. 
The inequality follows from that every state in $\Omega_1$ can be written as a combination of a state in $\Omega_0$ and a long DW, see \cref{fig:peierls}.
In particular, we have an injection of the following form
\begin{align}
    \Omega_1 \hookrightarrow \Omega_0 \times \{ \gamma : |\gamma| > \ell_\ast \}.
\end{align}

\begin{figure}[h]
    \centering
    \includegraphics[]{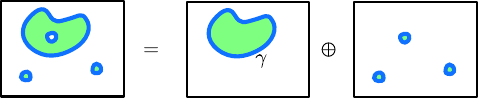}
    \caption{Peierls argument for Eq.~\eqref{eq:Ising_Peierls_series}.}
    \label{fig:peierls}
\end{figure}

The RHS of \cref{eq:Ising_Peierls_series} is a convergent geometric series for $y < y_\sharp$ (where $y_\sharp > 0$), in which case there exists constant $c > 0$, $L_0 < \infty$ so that the RHS is upper bounded by $\exp(-c L)$ for all $L > L_0$.
\end{proof}

It follows from \cref{thm:tmix_bottleneck} and \cref{thm:bottleneck_ratio_Ising_Peierls} that\footnote{We have not tried to make the constants explicit. We abuse the constant symbols $c$ and $C$ with the understanding that the constants here are local to the theorem, and are generally different from elsewhere.}
\begin{theorem}
\label{thm:Ising_exp_tmix}
The Ising model on $\Lambda = \mathbb{Z}_L \times \mathbb{Z}_L$ has an exponentially divergent mixing time at sufficiently low temperature, that is
\begin{align}
    y \in (0, y_\sharp) \  \Rightarrow \  \Phi_\ast \leq \exp(-c L) \  \Rightarrow \  t_{\rm mix} \geq \exp(c L).
\end{align}
\end{theorem}

Note that, in the presence of a longitudinal field ($J > 0$ or $x < 1$) we can no longer upper bound the bottleneck ratio as above.
First, the Ising symmetry is explicitly broken, and we have to consider $\Omega_\pm$ separately (rather than $\Omega_0$, as in \cref{eq:Ising_bottleneck_ratio_bound}).
Second, when we consider $\Omega_\pm$ that is anti-aligned with the field direction (the ``wrong bottle''), \cref{thm:bottleneck_ratio_Ising_Peierls} breaks down:
together with the factor $y^{|\contour|}$ in \cref{eq:Ising_Peierls_series}, there can be a nontrivial multiplicative factor $x^{-O(|M(\contour)|)}$ with $x = e^{-2\beta h} < 1$ and $M(\contour)$ an area of spins enclosed by $\contour$, whence the series is no longer convergent.
Physically, this is because the energy gain by flipping a domain of ``wrong'' spins can exceed the boundary energy cost, so that domain walls can no longer protect the wrong bottle to rapidly relax into the one favored by the field.
Indeed, whenever the field is nonzero ($x<1$) the Glauber dynamics mixes rapidly~\cite{martinelli1999saintflour},
\begin{theorem} \label{thm:Ising_rapid_tmix_field}
The Ising model on $\Lambda = \mathbb{Z}_L \times \mathbb{Z}_L$ mixes rapidly in the presence of a symmetry breaking field, at all temperatures,
\begin{align}
    y > 0, x < 1 \quad \Rightarrow \quad \tmix = O(\ln L).
\end{align}
\end{theorem}

\subsection{3D $\mathbb{Z}_2$ LGT without matter ($x = 1$)}

We now return to the 3D $\mathbb{Z}_2$ LGT, invoking the membrane representation in \cref{eq:Z_mem_expansion}.
In this section we focus on the case $x = 1$, and the more general case $x \lesssim 1$ is treated in \cref{sec:main_proof}.

With $x=1$, the expansion of $Z$ in \cref{eq:Z_mem_expansion} simplifies to a summation over disjoint collections $\contourset$ of closed ``flux loop'' configurations.
To see this, we rewrite the expansion as
\begin{align} \label{eq:Z_mem_expansion_rewritten}
    Z \propto \sum_{\contourset} y^{|\contourset|}  \cdot \sum_M \mathbb{1}_{\partial M = \contourset} \cdot x^{|M|}.
\end{align}
Indeed, as all membrane configurations are weighted by the same factor regardless of the area  for $x=1$ (i.e. membranes have exactly zero surface tension), the summation over membranes $M$ bounded by $\contourset$ contributes a constant factor, which can be omitted.
This is because the number of membranes with a given boundary does not depend on the boundary.

We define $\Omega_0, \decoder, \Omega_{\pm}, \Omega_1$ similarly as in \cref{eq:Ising_bottles,eq:Ising_decoder,eq:Ising_plus_minus_bottles_def,eq:Ising_Omega_1}, in the present case replacing ``domain walls'' with ``flux loops''.
We provide a few illustrations in \cref{fig:logical_states} to help visualize these notions.
Following the same line of arguments leading to \cref{thm:Ising_exp_tmix}, we have
\begin{theorem}
\label{thm:Ising_LGT_exp_tmix_x=1}
The $\mathbb{Z}_2$ LGT on $\Lambda = (\mathbb{Z}_L)^3$ with $x = 1$ has an exponentially divergent mixing time at sufficiently small $y$, that is
\begin{align}
    y \in (0, y_\sharp) \  \Rightarrow \  \Phi_\ast \leq \exp(-c L) \  \Rightarrow \  t_{\rm mix} \geq \exp(c L).
\end{align}
\end{theorem}
\noindent We note that this is a special case of \cref{thm:Ising_LGT_exp_tmix_x_general}.

\section{Mixing time at weak Higgs coupling ($x \in (x_\flat, 1]$) \label{sec:main_proof}}

In this section we prove \cref{thm:Ising_LGT_exp_tmix_x_general}.
That is, for all $x \in (x_\flat, 1]$ with some $x_\flat > 0$, we prove an exponential mixing time by applying a similar Peierls-type proof as in \cref{sec:bottleneck_examples_Peierls}.

We will adopt a proof strategy similar to that of \cref{thm:Ising_exp_tmix,thm:Ising_LGT_exp_tmix_x=1}, by identifying the bottleneck states as those with (exactly one) \textit{long} flux loop.
However, we note that our previous analyses make use of an Ising symmetry in crucial ways, but which is explicitly broken when $x < 1$, as we now discuss.
\begin{enumerate}[(i)]
\item
First, as a technical point, our formulation of the Peierls inequality follows most naturally from an injection from $\Omega_1 \to \Omega_0 \times \{\contour: |\contour| \geq \ell_\ast \}$ (see proof of~\cref{thm:bottleneck_ratio_Ising_Peierls}) 
without specifying whether the image is in either one bottle $\Omega_\pm$.
Such a specification necessarily requires one to consider details of the decoder $\decoder$ and presumably an injection $\partial \Omega_\pm \to \Omega_\pm \times \{\contour: |\contour| \geq \ell_\ast \}$, which can be considerably more involved.
Instead, we can directly relate $\mu(\Omega_1)/\mu(\Omega_0)$ to $\mu(\partial \Omega_\pm)/\mu(\Omega_\pm)$ -- as required by the definition of a bottleneck -- by Ising symmetry, see \cref{eq:decoder_condition_2,eq:Ising_bottleneck_ratio_bound}).
Such a symmetry between $\Omega_\pm$ is also present for the $\mathbb{Z}_2$ LGT when $x = 1$ (used implicitly in \cref{thm:Ising_LGT_exp_tmix_x=1}), but is explicitly broken whenever $x < 1$.
\item
More physically, consider the ways in which such a symmetry is broken for $x < 1$.
A first observation is that states in $\Omega_+$ appear to be favored by $x < 1$, as they naively have a smaller surface area, compare \cref{fig:logical_states}.
The situation is similar to the 2D Ising model in a longitudinal field (see \cref{thm:Ising_rapid_tmix_field}) which breaks the Ising symmetry and favors $\Omega_+$ over $\Omega_-$.
Consequently, it immediately leads to rapid mixing (therefore no small bottleneck ratio can be found).
\item
Therefore, even if the technical point in (i) can be avoided and the ratio between $\mu(\Omega_{0,1})$ can be useful for upper bounding $\Phi_\ast$, a Peierls series (similar to
\cref{eq:Ising_Peierls_series}) is still manifestly divergent due to factors such as $x^{-|M(\contour)|}$ where $\partial M = \contour$, compare discussions preceeding \cref{thm:Ising_rapid_tmix_field}.
\end{enumerate}
Based on these observations, it would appear that a Peierls-type proof cannot be found whenever $x < 1$.

To make progress, consider the thermodynamic phase diagram, \cref{fig:lattice}(d).
Unlike the 2D Ising model where a longitudinal field $x < 1$ is \textit{relevant}, in the 3D $\mathbb{Z}_2$ LGT a nonzero field $x < 1$ is \textit{irrelevant} in the deconfined phase.
Indeed, in 3D, the surface tension receives both \textit{energetic} and \textit{entropic} contributions, and this is the crucial difference between 2D and 3D~\cite{poulin-melko-hastings}.
Whenever there is a nonzero \textit{bare} surface tension for $x \lesssim 1$, one expects that it is \textit{renormalized} to zero at long length scales by \textit{entropy}, giving $x_{\rm eff} = 1$.
With this picture, the issues we raised in the previous paragraphs are resolved.
That is, the symmetry between the two bottles $\Omega_\pm$ should be restored at long length scales, and there should be a convergent Peierls series with $x_{\rm eff} = 1$.
A proof similar (in spirit) to that of \cref{thm:Ising_LGT_exp_tmix_x=1} should follow, where the long mixing time is guaranteed by a nonzero \textit{line tension} of flux loops.

In the following, we make these words precise.
With our proof we hope to convey that (1) the zero renormalized surface tension leads to a precise notion of a restored one-form symmetry between $\Omega_\pm$ (explicitly broken by $x \lesssim 1$),
and that (2) the slow mixing time is a consequence of the zero (renormalized) surface tension, which is a universal property of the deconfined phase.

We define $\Omega_{0,1}$ similarly as before,
\begin{align} \label{eq:Ising_LGT_Omega_01}
    \Omega_{m} \coloneqq
    \{M \in \Omega : N_{\rm LFL}(M) = m \}, 
    \quad m \in \{0, 1\}.
\end{align}
where $N_{\rm LFL}$ is the \textit{long (connected) flux loop (FL) count}
\begin{align}
    N_{\rm LFL}(M) \coloneqq |\{\alpha: \contour_\alpha \subset \partial M, |\contour_\alpha| \geq \ell_\ast \}|.
\end{align}
We also define 
\begin{align}
\label{eq:decoder_LGT}
    & \delta(\contour) \coloneqq \mathrm{argmin}_{M \in \Omega: \partial M = \contour} |M|, \\
    & \decoder M \coloneqq M \oplus \(\bigoplus_{\alpha} \delta (\contour_\alpha) \), \quad \text{ where } \partial M = \bigsqcup_{\alpha} \contour_\alpha
\end{align}
where $\decoder$ is a well-defined \textit{decoder} that brings any $M \in \Omega_0$ to a no-flux-loop state, as illustrated in~\cref{fig:LGT_decoder}.
The no-flux-loop states fall into two homology classes, as distinguished by the value of the Wilson loop operator in the $z$ direction.
This defines the \textit{decoder} $\decoder: \Omega_0 \to \{+1, -1\}$.\footnote{With periodic boundary conditions in all three directions on $(\mathbb{Z}_L)^3$, one can in principle encode three logical bits, using the value of the Wilson loop in the $x, y$, and $z$ directions.
For simplicity, throughout the section we will focus only on the logical bit encoded by the Wilson loop in the $z$ direction, with the understanding that our analysis generalizes straightforwardly to the case of three logical bits.}
The decoder defines the bottles
\begin{align}
    \label{eq:Ising_LGT_plus_minus_bottles_def}
    \Omega_\pm \coloneqq \decoder^{-1}(\pm 1).
\end{align}
We also define
\begin{align}
    \Omega_{W} \coloneqq \mathrm{argmin}_{\Omega_W \in \{\Omega_+, \Omega_-\}}\, \mu(\Omega_W).
\end{align}
By definition, 
\begin{align} \label{eq:mu_Omega_min_below_half}
    \mu(\Omega_{W}) \leq 1/2 \cdot \mu(\Omega_0) \leq 1/2.
\end{align}
We have, similarly to~\cref{prop:Ising_bottleneck_condition},
\begin{proposition}\label{prop:Ising_LGT_bottleneck_condition}
The bottles $\Omega_\pm$ as defined in \cref{eq:Ising_LGT_plus_minus_bottles_def} satisfy the following condition under Glauber dynamics $(\Omega, P)$
\begin{align}
    \label{eq:decoder_condition_1}
    \forall \state \in \Omega_+, \quad P(\state,\state') > 0 
    \ \Rightarrow\ 
    \state' \in \Omega_+ \cup \Omega_1.
\end{align}
A similar results holds for $\state \in \Omega_-$.
Comparing \cref{eq:def_bottleneck}, it follows that
\begin{align}
    \partial \Omega_\pm \subset \Omega_1 \quad \text{and} \quad \mu(\partial \Omega_\pm) \leq \mu(\Omega_1).
\end{align}
\end{proposition}

\begin{figure}
    \centering
    \includegraphics[]{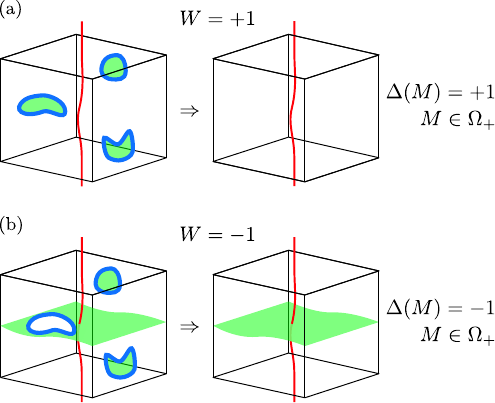}
    \caption{Illustration of the LGT decoder defined in \cref{eq:decoder_LGT}}
    \label{fig:LGT_decoder}
\end{figure}

\subsection{Proof of \cref{thm:Ising_LGT_exp_tmix_x_general}}

With these preparations, we detail below
\begin{proof}[Proof of~\cref{thm:Ising_LGT_exp_tmix_x_general}]

By \cref{prop:Ising_LGT_bottleneck_condition}, \cref{eq:mu_Omega_min_below_half}, and \cref{eq:def_phi_ast,eq:def_phi_A,eq:def_bottleneck,eq:phi_A_bottleneck}, 
\begin{align}
    \Phi_\ast \leq \frac{\mu(\partial \Omega_{W})}{\mu(\Omega_{W})} \leq \frac{\mu(\Omega_1)}{\mu(\Omega_{W})} \eqqcolon \Phi_W.
\end{align}
We rewrite
\begin{align}
    & \Phi_W
    =
    \frac{\sum_{\contourset} \mathbb{1}_{N_{\rm LFL}(\contourset) = 1} \sum_{M} \mathbb{1}_{\partial M = \contourset} \, x^{|M|} y^{|\contourset|}}
    {\sum_{\contourset} \mathbb{1}_{N_{\rm LFL}(\contourset) = 0} \sum_{M} \mathbb{1}_{\partial M = \contourset} \cdot \mathbb{1}_{\decoder(M) = W} \, x^{|M|} y^{|\contourset|}}.
\end{align}
We define \textit{conditional ensembles} as follows,
\begin{align}
    \label{eq:def_conditional_ensemble}
    & Z(x, \contourset) \coloneqq \sum_M \mathbb{1}_{\partial M = \contourset} \cdot x^{|M|},\\
    \label{eq:def_conditional_ensemble_2}
    & Z_{\pm}(x, \contourset) \coloneqq \sum_M \mathbb{1}_{\partial M = \contourset} \cdot \mathbb{1}_{\decoder(M)=\pm} \cdot x^{|M|}.
\end{align}
With these, we can write
\begin{align}
    & \Phi_W
    =
    \frac{\sum_{\contourset} \mathbb{1}_{N_{\rm LFL}(\contourset) = 1} \cdot y^{|\contourset|} \cdot Z(x,\contourset)}
    {\sum_{\contourset} \mathbb{1}_{N_{\rm LFL}(\contourset) = 0} \cdot y^{|\contourset|} \cdot Z_W(x, \contourset)}.
\end{align}
Using the injectivity arguments from Peierls (see proof of~\cref{thm:bottleneck_ratio_Ising_Peierls}), we have
\begin{align}
    & \sum_{\contourset} \mathbb{1}_{N_{\rm LFL}(\contourset) = 1} \cdot y^{|\contourset|} \cdot Z(x, \contourset) \nn
    \leq & 
    \sum_{\text{contour $\contour$}} \mathbb{1}_{|\contour| \geq \ell_\ast} 
    \cdot
    y^{|\contour|} 
    \sum_{\contourset} \mathbb{1}_{N_{\rm LFL}(\contourset) = 0}
    \cdot
    y^{|\contourset|}
    \cdot
    Z(x, \contourset \oplus \contour).
\end{align}
Thus,
\begin{align}
    \label{eq:phi_W_after_peierls}
    & \Phi_W 
    \leq 
    \sum_{\text{contour $\contour$}} \mathbb{1}_{|\contour| \geq \ell_\ast} \cdot y^{|\contour|} \cdot \nn
    &\hspace{.5in} 
    \(
    \frac{\sum_{\contourset} \mathbb{1}_{N_{\rm LFL}(\contourset) = 0}\cdot y^{|\contourset|} \cdot Z(x, \contourset \oplus \contour)}
    {\sum_{\contourset} \mathbb{1}_{N_{\rm LFL}(\contourset) = 0}\cdot y^{|\contourset|} \cdot Z_W(x, \contourset)}
    \).
\end{align}
By \cref{lemma:bottle_symmetry} and \cref{lemma:wilson_loop_perimeter_lower_bound} (stated and proven below), there exists $x_\flat > 0$ such that for sufficiently large $L$
\begin{align}
    & \forall x \in (x_\flat, 1], \exists C(x) < \infty \text{ s.t. } \forall \contourset, \contour, \nn
    & \hspace{.1in}
    (0.499) \cdot
    \frac{Z(x, \contourset \oplus \contour)}{Z_W(x, \contourset)}
    \leq 
    \frac{Z(x, \contourset \oplus \contour)}{Z(x, \contourset)}
    \leq \exp{C(x) |\contour|}.
\end{align}
Here, we use \cref{lemma:bottle_symmetry} for the first inequality, and \cref{lemma:wilson_loop_perimeter_lower_bound} for the second.
This puts an uniform upper bound on the term in parentheses of \cref{eq:phi_W_after_peierls}.
Therefore,
\begin{align}
    \Phi_W \leq (\mathrm{const}) \sum_{\text{contour $\contour$}} \mathbb{1}_{|\contour| \geq \ell_\ast} 
    \left( y \cdot e^{C(x)} \right) ^{|\contour|}.
\end{align}
Summing the geometric series we have that $\forall y \in [0, y_\sharp(x))$ where $y_\sharp(x)$ depends on $C(x)$ and the coordination number of $\Lambda$, $\Phi_W \leq \exp{-c'(x,y) \cdot  \ell_\ast}$ where $c'(x,y) > 0$ for all $x \in (x_\flat, 1], y \in [0, y_\sharp(x))$.
\end{proof}

\subsection{Proof of~\cref{lemma:bottle_symmetry}}

\begin{proof}[Proof of \cref{lemma:bottle_symmetry}]
Consider the relative difference of the two conditional bottles,
\begin{align}
    \mathcal{O}(x, \contourset) \coloneqq \frac{Z_+(x, \contourset) - Z_-(x, \contourset)}{Z_+(x, \contourset) + Z_-(x, \contourset)}.
\end{align}
Here it is possible that $\contourset = \emptyset$;
recall that the decoder $\decoder$ acts trivially when the input $M$ has $\partial M = \emptyset$.
Note that
\begin{align}
    \frac{Z_W(x, \contourset)}{Z(x, \contourset)}  
    \geq
    \frac{\mathrm{min}\{Z_\pm(x, \contourset)\}}{Z(x, \contourset)}  
    = \frac{1}{2}(1 - |\mathcal{O}(x,\contourset)|).
\end{align}
So, an upper bound on $|\mathcal{O}(x,\contourset)|$ will be sufficient for our purpose.
Below, we will use a dual representation for this upper bound.

For fixed $x, \contourset$, we define a 3D Ising model with spins on the vertices of $\Lambda = (\mathbb{Z}_L)^3$.
Compared to the usual ferromagnetic Ising model, 
we take a quenched disorder realization where on certain bonds (defined by a membrane $M$ bounded by $\contourset$) the Ising coupling is taken to be antiferromagnetic.
We also introduce an additional, auxiliary Ising spin, denoted $\state_{\rm aux}$.
The partition function reads
\begin{align} \label{eq:dual_3D_Ising_AFM}
    & \ \Xi^{\rm I}(x, \contourset) \nn
    & = \tr_{\state, \state_{\rm aux}} \exp \ld \beta \sum_{\avg{ij}} (-1)^{\mathbb{1}_{\avg{ij} \in M}} \cdot \(\state_{\rm aux}\)^{\mathbb{1}_{\avg{ij} \in P}} \cdot \state_i \state_j \rd,
\end{align}
where $e^{-2\beta} = x$, $M$ is any membrane such that $\partial M = \contourset$, and in addition $P$ is a homologically nontrivial membrane.
Therefore, $\sigma_{\rm aux} = +1$ enforces \textit{periodic} boundary conditions in the dual Ising model, and $\sigma_{\rm aux} = -1$ enforces \textit{antiperiodic} boundary conditions.
An illustration of the model is shown in \cref{fig:dual-3D-Ising}.
When $\contourset = \emptyset$,  we take $M = \emptyset$, and $\Xi^{\rm I}(x, \contourset)$ simplifies to
\begin{align} \label{eq:dual_3D_Ising}
    \Xi^{\rm I}(x, \emptyset) = \tr_{\state, \state_{\rm aux}} \exp \ld \beta \sum_{\avg{ij}} \(\state_{\rm aux}\)^{\mathbb{1}_{\avg{ij} \in P}} \cdot \state_i \state_j \rd.
\end{align}

\begin{figure}[h]
    \centering
    \includegraphics[width=.20\textwidth]{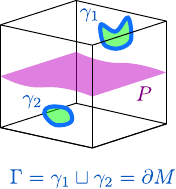}
    \caption{Illustration of the dual Ising model defined in \cref{eq:dual_3D_Ising_AFM}.
    }
    \label{fig:dual-3D-Ising}
\end{figure}

We observe that 
$\Xi^{\rm I}(x, \contourset)$ and $Z(x, \contourset)$ are related to each other by a duality transformation for all $\contourset$
\begin{align}
    \Xi^{\rm I}(x, \contourset) \propto Z(x, \contourset).
\end{align}
This follows from viewing \cref{eq:def_conditional_ensemble} as a ``low temperature expansion'' of \cref{eq:dual_3D_Ising_AFM}, where the membrane configurations for $Z(x,\contourset)$ are identified as domain walls for $\Xi^{\rm I}(x,\contourset)$.
Note that homologically distinct membranes can appear in \cref{eq:def_conditional_ensemble}, and this is accounted for by having $\state_{\rm anc}$ in \cref{eq:dual_3D_Ising_AFM}, which effectively introduces periodic and antiperiodic boundary conditions in the Ising model.

With this duality we note that
\begin{align}
    \label{eq:sigma_aux_Ising_empty}
    \mathcal{O}(x, \emptyset) =
    \avg{\state_{\rm aux}}^{\rm I}_{x, \emptyset},
\end{align}
where $\avg{\cdot}^{\rm I}_{x,\emptyset}$ is taken with respect to the ensemble $\Xi^{\rm I}(x,\emptyset)$.
We know $\avg{\state_{\rm aux}}^{\rm I}_{x, \emptyset} > 0$ by a Griffiths-Kelly-Sherman (GKS) inequality \cref{eq:GKS_1} applied to $\Xi^{\rm I}(x,\emptyset)$, which has all  couplings being ferromagnetic.
More generally, for an arbitrary $\Gamma$,
\begin{align}
    |\mathcal{O}(x, \contourset)| =
    |\avg{\state_{\rm aux}}^{\rm I}_{x, \contourset}|.
\end{align}

Next, we note that for an arbitrary $\contourset$
\begin{align}
    \left| \mathcal{O}(x, \contourset) \right|
    =
    \left| \avg{\state_{\rm aux}}^{\rm I}_{x, \contourset} \right|
    \leq
    \avg{\state_{\rm aux}}^{\rm I}_{x, \emptyset}
    = \mathcal{O}(x, \emptyset).
\end{align}
This follows from another GKS inequality \cref{eq:GKS_3}, allowing a comparison between $\Xi^{\rm I}(x, \contourset)$ and $\Xi^{\rm I}(x, \emptyset)$.

Finally, we have for sufficiently large $L$
\begin{align}
    \label{eq:Ising_apbc_F_upper_bound}
    \left| \mathcal{O}(x, \contourset) \right|
    \leq
    \avg{\state_{\rm aux}}^{\rm I}_{x, \emptyset} \leq e^{-c L}
\end{align}
for some $c > 0$ whenever $x > x_\flat$.
The last inequality follows from an explicit cluster expansion for $x > x_\flat$, see \cref{sec:cluster_expansion}.
\end{proof}

\subsection{Statement and proof of \cref{lemma:wilson_loop_perimeter_lower_bound}}

\begin{lemma}
\label{lemma:wilson_loop_perimeter_lower_bound}
For all $x \in (x_\flat, 1]$ and sufficiently large $L$, $\exists C(x) < \infty$ such that
\begin{align}
    \forall \contourset, \contour, \quad
    \frac{Z(x, \contourset \oplus \contour)}{Z(x, \contourset)}
    \leq \exp{C(x) |\contour|}.
\end{align}
\end{lemma}
\begin{proof}
We first observe that, by another duality, the conditional ensemble $Z(x, \contourset)$ takes the form of a Wilson loop \textit{in a pure gauge theory},
\begin{align}
    \label{eq:ratio_conditional_ensemble_Wilson_loop}
    \frac{Z(x,\contourset)}{Z(x, \emptyset)} = \avg{W(\contourset)}_x^{\rm G},
\end{align}
where the expectation value is taken with respect to the following model
\begin{align}
    \label{eq:dual_pure_LGT}
    \Xi^{\rm G}(x) = \sum_{\state} \exp(\beta^\ast \sum_{\widetilde{\Box}_{\widetilde{ijkl}}} \state_{\widetilde{ij}} \state_{\widetilde{jk}} \state_{\widetilde{kl}} \state_{\widetilde{li}}), \quad \tanh \beta^\ast = x.
\end{align}
Note that the model differ from that in \cref{eq:Z_def}, in that the spins now live on \textit{dual} edges of the form $\avg{\widetilde{ij}}$.
\cref{eq:ratio_conditional_ensemble_Wilson_loop} follows from another Kramers-Wannier duality,
\begin{align}
    \Xi^{\rm G}(x) \propto Z(x, \emptyset).
\end{align}

Therefore,
\begin{align}
    \frac{Z(x, \contourset \oplus \contour)}{Z(x, \contourset)}
    = \frac{\avg{W(\contourset \oplus \contour)}_{x}^{\rm G}}{\avg{W(\contourset)}_x^{\rm G}}.
\end{align}
Applying the GKS inequality \cref{eq:GKS_2} to $W(\contourset) = W(\contourset \oplus \contour) W(\contour)$ with respect to $\Xi^{\rm G}(x)$, we have
\begin{align}
    \label{eq:wilson_loop_upper_bound}
    \frac{\avg{W(\contourset \oplus \contour)}_x^{\rm G}}{\avg{W(\contourset)}_x^{\rm G}} \leq \frac{1}{\avg{W(\contour)}_x^{\rm G}} \leq \frac{1}{\exp{-C(x)|\contour|}}.
\end{align}
The last inequality follows from a ``perimeter law'' lower bound for the Wilson loop $\avg{W(\contour)}_x^{\rm G}$ for sufficiently large $x$ and arbitrary $\contour$.
This result should be well established, 
and we provide an explicit calculation in \cref{sec:cluster_expansion}.
\end{proof}

\section{Arrhenius law and critical mixing times \label{sec:critical_mixing_times}}

The bottleneck theorem (\cref{thm:tmix_bottleneck}) makes precise intuitions provided by the Arrhenius law for thermally activated processes~\cite{KRAMERS1940284, LANGER1969258}, that is, the \textit{escape rate} from a potential well is
\begin{align}
    \label{eq:arrhenius}
    t_{\rm esc}^{-1} \propto e^{-\Delta F},
\end{align}
where $\Delta F$ is the free energy barrier.
Explicit calculations can be performed in the path integral formalism in terms of ``instantons'' for single-particle systems, or for many-body systems where the phase space can be parameterized by a single variable (see e.g.~\cite{LANGER1969258, polyakov1987gauge, cugli2025introduction}).
For general many-body systems, however, such calculations seem to be quite difficult.

With such intuitions, we discuss (heuristically) aspects of the Glauber dynamics that are not captured by the equilibrium state.
Consider the two phase transitions out of the deconfined phase, either by increasing $T$ (confinement transition), or by increasing $h$ (Higgs transition).
The two thermodynamic transitions are both second order, and both have 3D Ising critical exponents; indeed, they are related by an exact electro-magnetic duality.
However, the critical mixing times exhibit markedly different scaling at the two transitions, as we now discuss.

\begin{itemize}

\item

\emph{Higgs transition} ---
Across the Higgs transition, the flux loops maintain a nonzero line tension (small $T$), whereas the bare surface tension increases above a critical value [from $h < h_c(T)$ to $h > h_c(T)$].
Above the Higgs transition ($h > h_c(T)$), for an initial state within $\Omega_-$ to relax into $\Omega_+$ (thereby losing the initial memory), the free energy barrier is set by a flux loop whose length $\ell_\ast$ is given by 
\begin{align}
    \sigma_{\rm eff} \cdot (\ell_\ast)^2 \approx \lambda_{\rm eff} \cdot \ell_\ast \quad \Leftrightarrow \quad \ell_\ast \approx \sigma_{\rm eff}^{-1} \cdot \lambda_{\rm eff},
\end{align}
where $\sigma_{\rm eff} > 0$ is the effective surface tension, and $\lambda_{\rm eff}$ the effective line tension.
That is, once the free energy barrier for creating a flux loop of size $\geq \ell_\ast$ is overcome, the system can lower its free energy by further expanding the loop, until the memory is lost.
This discussion is in parallel with that of 2D Ising model in a nonzero magnetic field, see~\cref{thm:Ising_rapid_tmix_field}.

When approaching the Higgs transition from the Higgs phase, we expect $\ell_\ast$ to diverge, and correspondingly $\sigma_{\rm eff} \to 0$.
Assuming the correlation length $\xi$ is the only length scale in the system, we expect $\ell_\ast \propto \xi$ when close to the critical point, and $\ell_\ast \propto L$ when $\xi \gg L$.
From \cref{eq:arrhenius} we have that at the Higgs transition
\begin{align}
    \label{eq:t_mix_higgs}
    t^{\rm H}_{\rm mix}(h=h_c(T), T) \propto \exp( \Delta F), \  \Delta F \propto \lambda_{\rm eff} \cdot L.
\end{align}
Therefore, at the Higgs transition, we expect the system to be slow mixing with mixing time $\ln t_{\rm mix} = \Omega(L)$, as in the deconfined phase.

\item

\emph{Confinement transition} ---
Instead, at the confinement transition, the loss of memory is due to the vanishing of $\lambda_{\rm eff}$, while $\sigma_{\rm eff} = 0$ throughout.
Analogies with the 2D Ising model at $T = T_c$ can be made, and we expect a critical slowdown~\cite{hohenberghalperin1977, slylubetzky2012critical} of the form
\begin{align}
    \label{eq:t_mix_confinement}
    t^{\rm C}_{\rm mix}(h, T=T_c(h)) \propto L^z,
\end{align}
where $z$ is the dynamic exponent.

\end{itemize}

We test these predictions with a numerical memory experiment.
We initialize the system in $\Omega_-$ as in \cref{fig:logical_states}(b), obtained by flipping a single noncontractible membrane on the all-0 state $\state = \mathbf{0} \in \Omega_+$.
We then perform Glauber updates according to the Metropolis rule \cref{eq:metropolis} to evolve the system.
We use a combination of single-spin updates and block updates of all six spins around a given vertex,
in order to mix the system slightly faster.\footnote{These update rules differ slightly from those defined in \cref{sec:glauber-metropolis} due to the additional six-spin flip move. However, such a move does not change the flux loop configuration, and thus the same bottles and bottlenecks in \cref{sec:main_proof} (which are only defined through flux loop configurations) also apply to this dynamics.
Therefore, this dynamics mixes exponentially slowly within the range of applicability of \cref{thm:Ising_LGT_exp_tmix_x_general}.
More generally, we expect no qualitative differences with single-spin-flip dynamics, throughout the phase diagram.}
By convention, a single Glauber step consists of $3L^3$ attempted single-spin updates and $L^3$ attempted block updates.

To extract the memory time we 
decode the system by 
running a variant of the sweep decoder~\cite{kubica2024PCA}.
Similarly to the decoder $\decoder$ defined in \cref{sec:main_proof} and \cref{fig:ising_decoder}, the sweep decoder takes a state $M \in \Omega$ to a no-flux-loop state, such as those in \cref{fig:logical_states}.
We say the memory fails whenever the decoder returns a state in $\Omega_+$.
We fit the failure probability
to an exponential decay 
\begin{equation}
\mathbb{P}_{\rm fail}(t) = C \cdot \exp (-t/t_\textrm{mem}) + \mathbb{P}_{\rm fail}(t\to\infty).
\end{equation}
We use the extracted memory time $t_{\rm mem}$ of the Wilson loop as a lower bound for the mixing time $\tmix$.

We first find the location of the phase transitions ($T_c \approx 1.31$ at a fixed $h = 0.2$ for the confinement transition, and $h_c \approx 0.246$ at a fixed $T = 1.1$ for the Higgs transition) from singularities of thermodynamic observables (in particular, the specific heat), and use these values for our data collapse below.
We find that dynamical phase transitions occur at the same location as the thermodynamic ones.

\begin{figure*}
    \includegraphics[width=.95\textwidth]{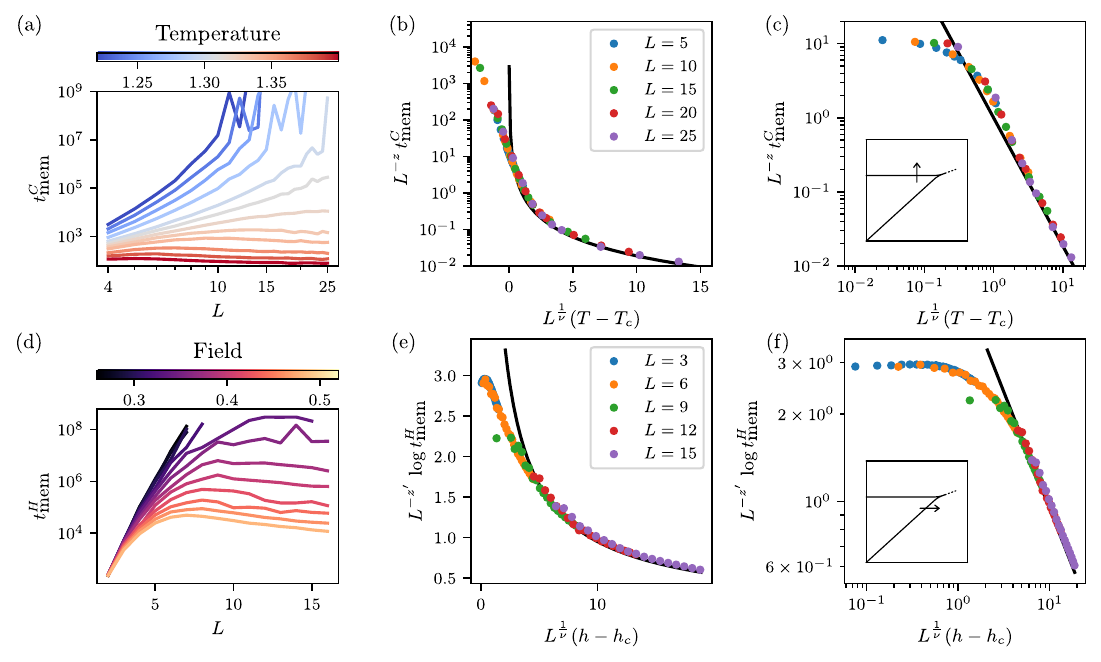}
    \caption{Raw data and data collapses at the two transitions. 
    (a,b,c) Data near the confinement (temperature-driven) transition.
    (a) Raw memory times on a log-log plot, so that the data at criticality is roughly linear. Note the slight upward bend indicating that $T=1.3$ is below the transition. 
    (b) Data collapse according to \cref{eq:tmem_scaling_form_confinement}, with parameters $\nu=0.64$, $z=2.7$, and $T_c=1.313$.
    The black line is the asymptotic behavior $\Psi(w) \propto w^{-\nu z}$, see \cref{eq:Psi_W_asymptotics_confinement}.
    (c) The same collapsed data as in (b) but on a log-log plot to highlight the algebraic decay at large $L^{1/\nu}(T-T_c)$. 
    (d,e,f) Data near the Higgs (field-driven) transition.
    (d) Raw memory times on a semilog plot. Note that each field value peaks and then decays, different from the confinement regime. 
    (e) Data collapse according to \cref{eq:t_mix_H_collapse} (note the log in the $y$ axis) with parameters $h_c=0.246$, $z'=0.97$, $\nu=0.64$.
    The black line is the asymptotic behavior $\Phi(w) \propto w^{-\nu'}$ with  $\nu' \approx 0.8$, compare \cref{eq:Phi_W_asymptotics_Higgs}.
    The discrepancy between $\nu$ and $\nu'$ is discussed in the main text.}
    \label{fig:transitions}
\end{figure*}

The different scaling of $t_{\rm mem}$ at the two transitions is supported by our numerical results.
As we show in \cref{fig:transitions}(a-c), across the confinement transition where we set $h = 0.2$ and vary $T$, the memory time is superpolynomial below the critical temperature and subpolynomial above the critical temperature.
The results can be collapsed against the following scaling function
\begin{align}
\label{eq:tmem_scaling_form_confinement}
t^{\rm C}_{\rm mem} = L^z \cdot \Psi[ (T-T_c) \cdot L^{1/\nu} ],
\end{align}
where we find $T_c\approx 1.31$, $\nu \approx 0.64$, and $z \approx 2.7$.
The critical exponent $\nu$ is consistent with that of the 3D Ising model, and $z$ is consistent with previous numerical results on the confinement transition~\cite{PhysRevB.97.024432}.\footnote{
We note that this dynamic exponent differs from that of the 3D Ising model.
Therefore, we have another example of two transitions with same static exponents (they are related again by duality) but different dynamic exponents.
For model A dynamics of the 3D Ising model, previous studies have found $z \approx 2.02$, see Ref.~\cite{PhysRevE.101.022126} and references therein.}
To match with $t^{\rm C}_{\rm mem}$ on the two sides of the transition, we expect the asymptotics of $\Psi$ to be
\begin{align}
    \label{eq:Psi_W_asymptotics_confinement}
    \Psi(w) \propto 
    \begin{cases}
        w^{-\nu z}, \quad & w \to +\infty \\
        \exp{|w|^\nu}, \quad & w \to -\infty
    \end{cases}.
\end{align}
The asymptotic behavior at $w \to +\infty$ is confirmed by our data, see \cref{fig:transitions}(c).

The results for the Higgs transition are shown in \cref{fig:transitions}(d-f), where we take a cut across the phase diagram along $T = 1.1$.
At the critical point, the memory time appears to be superpolynomial, see \cref{eq:t_mix_higgs}, as
in the memory phase.
However, outside of the memory phase the memory time still increases superpolynomially for small system sizes, and only
saturates
after reaching the correlation length scale $\xi$ that diverges at the transition.
We find the following data collapse 
\begin{align}
\label{eq:t_mix_H_collapse}
    \ln t^{\rm H}_{\rm mem} = L \cdot \Phi[ (h-h_c) \cdot L^{1/\nu} ],
\end{align}
where $h_c = 0.246$, and the same 3D Ising correlation length exponent $\nu \approx 0.64$ is used.
Here, we expect
\begin{align}
    \label{eq:Phi_W_asymptotics_Higgs}
    \Phi(w) \propto 
    \begin{cases}
        w^{-\nu}, \quad & w \to +\infty \\
        \rm{const.}, \quad & w \to -\infty
    \end{cases}.
\end{align}
We note that the numerical asympotic powerlaw decay of $\Phi(w)$ at large positive $w$ seem to differ from $\nu$, as we show in \cref{fig:transitions}(f).
We attribute this to the non-monotonic (and presumably non-universal) dependence of $t_{\rm mem}^{\rm H}$ on $L$ at small $L$, which are still observable at small systems sizes, see \cref{fig:transitions}(d).
We have not been able to access $h < h_c$ due to the rapid divergence of the memory time with $L$.

In Ref.~\cite{poulin-melko-hastings}, it was observed that the memory phase appears to extend beyond the Higgs transition. 
The long memory time just outside the Higgs phase can therefore be attributed to rapid (exponential) dependence of $t_{\rm mem}$ on the correlation length $\xi$, as captured by \cref{eq:t_mix_H_collapse}.

We also performed the memory experiment at the multicritical point (where the Higgs and confinement transitions meet), and found a memory time that scales as $t_{\rm mem} \propto L^{2.5}$. 
This exponent is consistent with findings in Ref.~\cite{nahum2020selfdual}.
However, 
it is not clear to us whether the memory 
time of the Wilson loop is still the longest time scale in the system, and whether it provides a tight lower bound for
the mixing time.
We leave this question for future work.

\section{Discussions} \label{sec:discussion}

Our main result \cref{thm:Ising_LGT_exp_tmix_x_general} establishes the $\mathbb{Z}_2$ LGT as a simplest example of a finite-temperature classical self-correcting memory, stable to perturbations that explicitly break the one-form symmetry of the model.
We note that in physical systems one-form symmetries are often emergent and not present at microscopic length scales.
We illustrate in detail how \textit{thermal fluctuations} compete with the nonzero surface tension (which explicitly favors one bottle over the other), thereby ``renormalizing'' the surface tension to zero and restoring the symmetry, see \cref{lemma:bottle_symmetry}. 

Our result applies to an open subregion of the deconfined phase; however, our proof illustrates that slow mixing relies on the properties of zero membrane tension and nonzero flux loop tension, which are universal to the deconfined phase.
We thus expect slow mixing to hold in the entirety of this phase.
This expectation is supported by our numerical memory experiments in \cref{sec:critical_mixing_times}.
Proving this is an open problem we leave for future studies.

Our work can be extended in a few different directions.

First, we have used the Glauber dynamics as a concrete model for relaxation, and have shown the memory is robust to the perturbation of a longitudinal field.
It would be desirable to prove the robustness of the memory to a broader class of perturbations --- such as generic local terms to the Hamiltonian (including those that introduce quantum fluctuations, in which case recent generalizations of the classical bottleneck theorem will apply~\cite{gamarnik2024slow,rakovszky2024bottlenecks}),
as well as stochastic noise processes that break detailed balance --- thereby establishing the $\mathbb{Z}_2$ LGT as an ``absolutely stable'' memory.

Second, we have rigorously established an ``entropic'' mechanism for emergent symmetries and robust memories. 
It would be interesting to explore if the notion of emergent symmetry can be useful in proving self-correction in  certain models.
A longstanding question in this direction is the absolute stability of the 4D toric code beyond its unperturbed fixed point~\cite{DKLP2001topologicalQmemory, Alicki2010}.
It is also interesting to ask
what other symmetries can be emergent, and whether they can be leveraged to protect self-correcting (classical or quantum) memories in low dimensions.

Below, we discuss two concrete related questions that connect to the literature.

\subsection{Formulating conditions for stable memory}

It would be useful to have a set of necessary or sufficient conditions for the stability of memory, in similar spirits of the proof of perturbative stability of the spectrum of gapped Hamiltonians~\cite{Bravyi_2010, Bravyi_2011, yin-lucas-2025-ldpc, deroeck-2025-ldpc}.
Recent work~\cite{ma2025circuitbasedchatacterizationfinitetemperaturequantum} has shown that sufficiently strong decay of local correlations can serve as one set of sufficient conditions.

Our work suggests a different, and perhaps more explicit direction for a set of such conditions. We rely on redundancies in our Hamiltonian for its energy barrier, and on local symmetries in our Hamiltonian for the entropic effects that promote stability.
In the language of homological algebra, 3D $\mathbb{Z}_2$ LGT is described by a length-3 chain complex connecting 4 complexes~\cite{A-Yu-Kitaev_1997, bombin2006homological, bravyi2013homological, breuckmann2018phd}.
By convention, bits and checks are vector spaces described by the middle two complexes, respectively; the first complex describes local symmetries of the model, while the last complex describes the redundancies among the checks.
In this representation, the \textit{local} check redundancies guarantee that the flux loop excitations are closed, and the \textit{local} symmetries (i.e. gauge invariance) leads to emergent symmetries in the presence of a symmetry-breaking field.
Thus, it might be possible to extend our proof to any model built from a suitably defined length-3 chain complex.

In the $\mathbb{Z}_2$ LGT, the combination of tensionless membranes with tensionful flux loops only breaks the Gibbs state into a finite number of bottles, labeled by the values of the (corrected) Wilson loops. The reason for this is the $\mathbb{Z}_2$ value of the membranes: a pair of membranes can be created by a series of local moves. 
Reference~\onlinecite{Stahl2025} explores a model with similar tensionless membranes with tensionful boundaries, but with the membranes carrying a $\mathbb{Z}_2^m$ label. Two membranes of the same color can annihilate with each other, but two of different colors cannot.
The result is that the Gibbs state breaks into $\exp(L)$ bottles, labeled by patterns of colors. Whether this Gibbs state shattering can remain robust in the same way as $\mathbb{Z}_2$ LGT remains an open question.

\subsection{Consequences of emergent one-form symmetries beyond the deconfined phase}

We note that in order to prove \cref{lemma:bottle_symmetry}, no conditions on $y$ (fugacity of flux loops) needs to be imposed.
Therefore, \cref{lemma:bottle_symmetry} holds beyond the deconfined phase, and our result seems to separate the region outside the deconfined phase into two parts, one with an emergent symmetry and one without, see Fig.~\ref{fig:lattice}(c).
However, given that the two parts belong to the same thermodynamic phase, physical consequences of such a separation are unclear. Previous work has shown that Higgs phases are SPTs of higher-form symmetries, and that the Higgs and confined phases of $\mathbb{Z}_2$ LGT are separated by a boundary phase transition~\cite{verresen2024higgscondensatessymmetryprotectedtopological, thorngren2023condensates}. Our work then raises the question of whether the boundary phase transition coincides with the loss of the emergent symmetry.

Additionally, Refs.~\onlinecite{nahum2020selfdual, nahum2024patching} have shown that a notion of emergent one-form symmetry beyond the deconfined phase may have consequences to the success (or failure) of a \textit{patching algorithm} that constructs fictitious Ising order parameters within the model.
More recently, Ref.~\onlinecite{liu2025detecting} made further connections of an emergent symmetry with a quantum error correction algorithm, when formulated within the quantum toric code model in 2D.
We expect that conditional bottles defined in \cref{eq:def_conditional_ensemble,eq:def_conditional_ensemble_2} and relaxed versions of \cref{lemma:bottle_symmetry} might provide a complementary tool for further explorations of emergent symmetries for this model and related ones.

Inside the deconfined phase, there are two Gibbs state components with equal weight, separated by a bottleneck. The existence of the bottleneck is decoder independent, so the notion of emergent one-form symmetry is also decoder independent. 
Outside of the phase, we still say that there are two sectors with equal weight, but the sectors are only defined by the decoder, but not by any bottleneck. Thus, we can still say that there is a one-form symmetry here, but this notion becomes decoder dependent.

\section*{Acknowledgements}

We thank Ethan Lake, Jeongwan Haah, Shengqi Sang, Nicholas O'Dea,   Adam Nahum, Zhehao Zhang, Ruochen Ma, Ruihua Fan, Sagar Vijay, Matthew Fisher, Dominik Hangleiter, Nathaniel Selub, Yunchao Liu, Tibor Rakovszky, Curt von Keyserlingk, and David Huse for helpful discussions.
We also thank Sourav Chatterjee for informative discussions and comments.

C.S. and V.K. are supported by the Office of Naval Research Young Investigator Program (ONR YIP) under Award Number N000142412098.
V.K. also acknowledges support from the
Packard Foundation through a Packard Fellowship in Science and Engineering.
B.P. acknowledges funding
through a Leverhulme-Peierls Fellowship at the University of Oxford and the Alexander von Humboldt foundation through a Feodor-Lynen fellowship.
Y.L. was supported in part by the Gordon and Betty Moore Foundation's EPiQS Initiative through
Grant GBMF8686, in part by the US Department of Energy, Office of Science under Award No DE-SC0019380, and in part by a Q-FARM Bloch Postdoctoral Fellowship at Stanford University.
We acknowledge the hospitality of the Kavli Institute for Theoretical Physics (KITP), which is supported in part by grant NSF PHY-2309135.

\bibliography{main}

\clearpage

\appendix

\begin{widetext}

\section{Recap -- GKS inequalities \label{sec:GKS}}

Here we collect a few Griffiths-Kelly-Sherman inequalities~\cite{griffiths1967a_correlations, griffiths1967b_correlations, griffiths1967c_correlations, kelly1968general} that we used multiple times in \cref{sec:main_proof}.
For a pedagogical introduction see e.g. Ref.~\cite{friedli_velenik_2017}.

\begin{apptheorem}[Theorem 3.49 of~\cite{friedli_velenik_2017}]
\label{thm:GKS}
Let $\mathcal{G}$ be any graph.
To each vertex $j \in \mathcal{G}$ we associate an Ising spin $\state_j$, and to each subset $C \subseteq \mathcal{G}$ we associate a non-negative ``coupling constant'' $K_C \ge 0$.
Consider the following model of Ising spins
\begin{align}
    Z_K \coloneqq \tr_{\state} \exp \left[ \sum_C K_C \state_C \right], \nonumber
\end{align}
where we define $\state_C = \prod_{j \in C} \state_j$.
Then we have 
\begin{align}
\label{eq:GKS_1}
    & \forall A \subseteq \mathcal{G}, \quad \langle \state_A \rangle_{Z_K} \ge 0, \\
\label{eq:GKS_2}
    & \forall A, B \subseteq \mathcal{G}, \quad \langle \state_A \state_B \rangle_{Z_K} \ge \avg{\state_A}_{Z_K} \avg{\state_B}_{Z_K}.
\end{align}
We note that here $A, B$ need not be disjoint.
\end{apptheorem}

~

\begin{appcorollary}[Exercise 3.31 of~\cite{friedli_velenik_2017}]
Let $K$ and $K'$ be two sets of coupling constants satisfying $K_C \geq |K'_C|$ for all $C \subseteq \mathcal{G}$.
We have
\begin{align}
\label{eq:GKS_3}
    \forall A \subseteq \mathcal{G}, \quad \langle \state_A \rangle_{Z_K} \ge \left| \langle \state_A \rangle_{Z_{K'}} \right|.
\end{align}
\end{appcorollary}

\section{Explicit cluster expansion for \cref{lemma:bottle_symmetry,lemma:wilson_loop_perimeter_lower_bound}} \label{sec:cluster_expansion}

Here we provide explicit bounds in \cref{eq:Ising_apbc_F_upper_bound,eq:wilson_loop_upper_bound}, following the standard methods of cluster expansion, see e.g. Chapter 5 of~\cite{friedli_velenik_2017}.

\subsection{Recap of the convergence condition}

Consider the following partition functions defined by a general membrane (possibly with boundaries) $P$,
\begin{align}
    \label{eq:Z_P_pm_appB}
    Z_{P, \pm} = \tr_{\state} \exp \ld \beta \sum_{\avg{ij}} \(\pm 1\)^{\mathbb{1}_{\avg{ij} \in P}} \cdot \state_i \state_j \rd, \quad \text{where $e^{-2\beta} = x$}.
\end{align}
We have that
\begin{align}
    \label{eq:z_two_cases_def}
    \mathcal{Z} \coloneqq \frac{Z_{P, -}}{Z_{P, +}} = \begin{cases}
        (1 - \avg{\state_{\rm aux}}^{\rm I}_{x, \emptyset})/(1 + \avg{\state_{\rm aux}}^{\rm I}_{x, \emptyset}), \quad& 
        \text{if $P$ is a noncontractible, boundaryless (i.e. logical) membrane} \\
        \avg{W(\partial P)}_{x}^{\rm G}, \quad& \text{if $\partial P \neq \emptyset$}
    \end{cases}.
\end{align}
Here, $\avg{\state_{\rm aux}}^{\rm I}_{x, \emptyset}$ is an expectation value taken within the ensemble $\Xi^{\rm I}_{x, \emptyset}$ (see \cref{eq:dual_3D_Ising_AFM}), and $\avg{W(\partial P)}_{x}^{\rm G}$ a Wilson loop expectation value taken within the ensemble $\Xi^{\rm G}_x$ (see \cref{eq:dual_pure_LGT}).
To see this, a Kramers-Wannier duality between \cref{eq:Z_P_pm_appB} and $\Xi^{\rm G}_x$ is again used.

We rewrite $Z_{P, \pm}$ in terms of the so-called \textit{polymer partition functions}~\cite{friedli_velenik_2017}, using a high-temperature expansion
\begin{align}
    & Z_{P, \pm} = 2^{|\Lambda|} \cosh(\beta)^{|\mathcal{E}_\Lambda|} \Upsilon_{\beta, P, \pm}, \nn
    &
    \Upsilon_{\beta, P, \pm} \coloneqq \sum_{\contourset' \subset \contourset}
    \ld\prod_{\contour \in \contourset'} w_\pm(\contour) \rd
    \ld \prod_{\{\contour, \contour'\} \in \contourset'} \delta(\contour, \contour') \rd,
\end{align}
and 
\begin{align}
    & \mathcal{E}_\Lambda = \{\text{links of the lattice $\Lambda$}\}, \\
    & \contourset = \{\contour \subset \mathcal{E}_\Lambda : \text{$\contour$ is connected, $\contour$ incidents each vertex an even number of times}\}, \\
    \label{eq:def_w_polymer}
    & w_\pm(\contour) = (\pm 1)^{|\contour \cap P|} (\tanh \beta)^{|\contour|}, \\
    & \delta(\contour, \contour') = \begin{cases}
        1, \quad &\text{if $\contour, \contour'$ have no vertex in common} \\
        0, \quad &\text{otherwise}
    \end{cases}.
\end{align}
We can rewrite
\begin{align}
    \Upsilon_{\beta, P, \pm} = 1 + \sum_{n \geq 1} \frac{1}{n!} \sum_{\contour_1 \in \contourset} 
    \cdots
    \sum_{\contour_n \in \contourset} 
    \ld \prod_{i=1}^n w_\pm(\contour_i) \rd
    \ld \prod_{1 \leq i < j \leq n} \delta(\contour_i, \contour_j) \rd.
\end{align}
Notice that only pairwise distinct polymer $n$-tuples $(\contour_1, \ldots, \contour_n)$ contribute to this summation, due to the interaction funtion $\delta$.
A \textit{formal} series expansion for the free energies can be written down,
\begin{align}
    \label{eq:log_xi_expansion}
    - F_{\beta, P, \pm} \coloneqq \ln \Upsilon_{\beta, P, \pm} = \sum_{m \geq 1} \sum_{\contour_1} \cdots \sum_{\contour_m}
    \varphi(\contour_1, \ldots, \contour_m) \prod_{i=1}^{m} w_\pm(\contour_i),
\end{align}
where 
\begin{align}
    & \varphi(\contour_1, \ldots, \contour_m) = \frac{1}{m!} \sum_{G \subset G_m} \mathbb{1}_{\text{$G$ is a connected subgraph of the complete graph $G_m$}} \prod_{\{i,j\} \in G} \zeta(\contour_i, \contour_j), \\
    & \zeta(\contour_i, \contour_j) \equiv \delta(\contour_i, \contour_j) - 1.
\end{align}

The series \cref{eq:log_xi_expansion} for $F_{\beta, P, \pm}$ is absolutely convergent provided~\cite{friedli_velenik_2017} that there exists $a : \contour \to \mathbb{R}_{> 0}$ such that 
\begin{align}
    \label{eq:condition_for_convergence}
    \forall \contour_\ast \in \contourset, \quad \sum_{\contour} |w_\pm(\contour)| |\zeta(\contour, \contour_\ast)| e^{a(\contour)} \leq a(\contour_\ast).
\end{align}
In the present case, we choose $a(\contour) = |\contour|$ and sufficiently small $\beta$ so that
\begin{align}
    \label{eq:condition_for_convergence_check}
    \mathrm{LHS} &= \sum_{\contour} \mathbb{1}_{\text{$\contour, \contour_\ast$ have vertices in common}} (\tanh \beta)^{|\contour|} e^{|\contour|} \nn
    &\leq |\contour_\ast| \max_{v \in \contour_\ast} \sum_{\contour \ni v}  (e \cdot \tanh \beta)^{|\contour|} \nn
    &= |\contour_\ast| \sum_{\ell \geq 1} \#\{\contour \ni 0: |\contour| = \ell \} \cdot (e \cdot \tanh \beta)^{\ell} \nn
    &= |\contour_\ast| \sum_{\ell \geq 1} ( (z-1) \cdot e \cdot \tanh \beta)^{\ell} \nn
    & \leq |\contour_\ast| = \mathrm{RHS}.
\end{align}
Notice that absolute convergence can be simultaneously established for both $F_{\beta, P, +}$ and $F_{\beta, P, -}$, as $|w_+(\contour)| = |w_-(\contour)|$.
It follows from \cref{eq:condition_for_convergence} that~\cite{friedli_velenik_2017}
\begin{align}
    \forall \contour_1 \in \contourset, \quad 1 + \sum_{m \geq 2} m \sum_{\contour_2} \ldots \sum_{\contour_m}  |\varphi(\contour_1, \ldots, \contour_m)| \prod_{i=2}^{m} |w_\pm(\contour_i)| \leq e^{a(\contour_1)},
\end{align}
from which the convergence of \cref{eq:log_xi_expansion} follows.

\subsection{Computing differences between free energies}

Given the absolute convergence of the series for both $F_{\beta, P, \pm}$, we can write
\begin{align}
    \label{eq:ln_z_diff}
    \ln \mathcal{Z} &= F_{\beta, P, +} -  F_{\beta, P, -} \nn
    &= \sum_{m \geq 1} \sum_{\contour_1} \cdots \sum_{\contour_m}
    \varphi(\contour_1, \ldots, \contour_m) \ld \prod_{i=1}^{m} w_-(\contour_i) - \prod_{i=1}^{m} w_+(\contour_i) \rd \nn
    &\coloneqq \sum_{m \geq 1} \sum_{X = (\contour_1, \ldots, \contour_m)} \ld \Psi_-(X) - \Psi_+(X) \rd,
\end{align}
where the last summation is over all ordered $m$-tuples $X$ with its elements from $\contourset$.

Comparing the definition of $w_\pm(\contour)$ in \cref{eq:def_w_polymer}, we see that the term $X = (\contour_1, \ldots, \contour_m)$ contributes to \cref{eq:ln_z_diff} if and only if $\sum_{i=1}^{m} |\contour_i \cap P|$ is an odd number.
Observe that
\begin{align}
    \mathbb{1}_{\text{$\sum_{i=1}^{m} |\contour_i \cap P|$ is odd}} \leq \sum_{i=1}^m \mathbb{1}_{\text{$|\contour_i \cap P|$ is odd}}.
\end{align}
With these, we may rewrite \cref{eq:ln_z_diff} as
\begin{align}
    \left| \ln \mathcal{Z} \right | &\leq \sum_{m \geq 1} \sum_{X = (\contour_1, \ldots, \contour_m)} \left| \Psi_-(X) - \Psi_+(X) \right| \nn
    &= \sum_{m \geq 1} \sum_{X = (\contour_1, \ldots, \contour_m)} 
    \mathbb{1}_{\text{$\sum_{i=1}^{m} |\contour_i \cap P|$ is odd}}
    \left| \Psi_-(X) - \Psi_+(X) \right| \nn
    &\leq \sum_{m \geq 1} \sum_{X = (\contour_1, \ldots, \contour_m)} 
    \ld \sum_{i=1}^m \mathbb{1}_{\text{$|\contour_i \cap P|$ is odd}} \rd \cdot
    \left| \Psi_-(X) - \Psi_+(X) \right| \nn
    &\leq 2 \cdot \sum_{m \geq 1} \sum_{X = (\contour_1, \ldots, \contour_m)} 
    \ld \sum_{i=1}^m \mathbb{1}_{\text{$|\contour_i \cap P|$ is odd}} \rd \cdot
    \left| \Psi_+(X) \right| \nn
    &\coloneqq 2 \cdot \Phi_P,
\end{align}
where
\begin{align}
    \label{eq:phi_P_convergence}
    \Phi_P &= \sum_{m \geq 1} \sum_{X = (\contour_1, \ldots, \contour_m)} 
    \ld \sum_{i=1}^m \mathbb{1}_{\text{$|\contour_i \cap P|$ is odd}} \rd \cdot
    \left| \Psi_+(X) \right| \nn
    &\leq \sum_{m \geq 1} m 
    \sum_{\contour_1} 
    \mathbb{1}_{\text{$|\contour_1 \cap P|$ is odd}}
    \sum_{\contour_2} \cdots \sum_{\contour_m}
    \left|\varphi(\contour_1, \ldots, \contour_m) \right| \cdot \prod_{i=1}^{m} \left|w_+(\contour_i)\right| \nn
    &\leq 
    \sum_{\contour_1} 
    \mathbb{1}_{\text{$|\contour_1 \cap P|$ is odd}} \cdot \left|w_+(\contour_1)\right| \cdot
    \ld
    1 + 
    \sum_{m \geq 2} m 
    \sum_{\contour_2} \cdots \sum_{\contour_m}
    \left|\varphi(\contour_1, \ldots, \contour_m) \right| \cdot \prod_{i=2}^{m} \left|w_+(\contour_i)\right|
    \rd \nn
    &\leq 
    \sum_{\contour_1} 
    \mathbb{1}_{\text{$|\contour_1 \cap P|$ is odd}}
    \cdot \left|w_+(\contour_1)\right| \cdot
    e^{|\contour_1|} \nn
    &\leq 
    \sum_{\contour_1} 
    \mathbb{1}_{\text{$|\contour_1 \cap P|$ is odd}}
    \cdot (e \cdot \tanh \beta)^{|\contour_1|}.
\end{align}
Since $\ln \mathcal{Z} \leq 0$, we have
\begin{align}
    \label{eq:bound_intermediate}
    e^{-2 \Phi_P} \leq \mathcal{Z} \leq 1.
\end{align}
Below, we upper bound $\Phi_P$ for the two cases in \cref{eq:z_two_cases_def}.

\begin{figure*}[b]
    \centering
    \includegraphics[width=.42\linewidth]{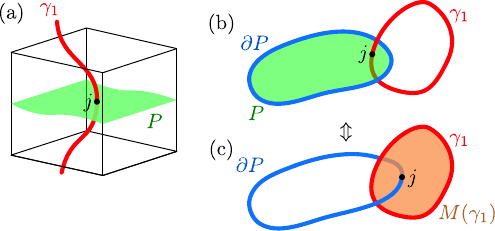}
    \caption{Terms contributing to the cluster expansion in (a) and \cref{eq:B19} (c,d) \cref{eq:B21}.
    }
    \label{fig:cluster_exp_config}
\end{figure*}

\subsection{When $P$ is a logical membrane -- for \cref{eq:Ising_apbc_F_upper_bound} of \cref{lemma:bottle_symmetry}}

Notice that on our lattice, for a logical membrane $P$,
\begin{align}
    \label{eq:B19}
    \mathbb{1}_{\text{$|\contour_1 \cap P|$ is odd}} \leq \mathbb{1}_{|\contour_1| \geq L} \cdot \sum_{j \in P} \mathbb{1}_{\contour_1 \ni j}.
\end{align}
For all $\beta$ satisfying the convergence conditions \cref{eq:condition_for_convergence,eq:condition_for_convergence_check}, we have from \cref{eq:phi_P_convergence}
\begin{align}
    \Phi_P & \leq \sum_{\contour_1} (e \cdot \tanh \beta)^{|\contour_1|} \cdot \mathbb{1}_{|\contour_1| \geq L} \cdot \sum_{j \in P} \mathbb{1}_{\contour_1 \ni j} \nn
    & \leq |P| \sum_{\contour_1} (e \cdot \tanh \beta)^{|\contour_1|} \cdot \mathbb{1}_{|\contour_1| \geq L} \cdot \mathbb{1}_{\contour_1 \ni 0} \nn
    & \leq (L^2) \sum_{\ell \geq L} \#\{\contour \ni 0: |\contour| = \ell \} \cdot (e \cdot \tanh \beta)^{\ell} \nn
    & \leq (L^2) \sum_{\ell \geq L} ( (z-1) \cdot e \cdot \tanh \beta)^{\ell} \nn
    & \leq e^{-c L}
\end{align}
with $c > 0$.
Comparing \cref{eq:z_two_cases_def,eq:bound_intermediate}, we have $\avg{\state_{\rm aux}}^{\rm I}_{x, \emptyset} \leq e^{-cL}$, and \cref{eq:Ising_apbc_F_upper_bound} holds for sufficiently large $L$.

\subsection{When $P$ is a membrane with boundary -- for \cref{eq:wilson_loop_upper_bound} of \cref{lemma:wilson_loop_perimeter_lower_bound}}

For each $\contour_1 \in \contourset$, we choose a membrane $M = M(\contour_1)$ such that $\partial M = \contour_1$ and $|M| \leq |\contour_1|^2$.
We observe that
\begin{align}
    \label{eq:B21}
    \mathbb{1}_{\text{$|\contour_1 \cap P|$ is odd}} = \mathbb{1}_{\text{$|\partial P \cap M(\contour_1)|$ is odd}} \leq \sum_{j \in \partial P} \mathbb{1}_{M(\contour_1) \ni j}.
\end{align}
For all $\beta$ satisfying the convergence conditions \cref{eq:condition_for_convergence,eq:condition_for_convergence_check}, we have from \cref{eq:phi_P_convergence}
\begin{align}
    \Phi_P & \leq \sum_{\contour_1} (e \cdot \tanh \beta)^{|\contour_1|} \sum_{j \in \partial P} \mathbb{1}_{M(\contour_1) \ni j} \nn
    & \leq \sum_{j \in \partial P} \sum_{\ell \geq 4}
    \#\{\contour: M(\contour) \ni j, |\contour| = \ell \} \cdot (e \cdot \tanh \beta)^{\ell} \nn
    & \leq \sum_{j \in \partial P} \sum_{\ell \geq 4}
    \#\{\contour: \contour \ni 0, |\contour| = \ell \} \cdot \ell^2 \cdot (e \cdot \tanh \beta)^{\ell} \nn
    & \leq |\partial P| \sum_{\ell \geq 4} \ell^2 \cdot ( (z-1) \cdot e \cdot \tanh \beta)^{\ell} \nn
    & \leq C |\partial P|,
\end{align}
with $C < \infty$.
Comparing \cref{eq:z_two_cases_def,eq:bound_intermediate}, we see that $\avg{W(\partial P)}_{x}^{\rm G} \geq \exp{-C |\partial P|}$, which is what we need in \cref{eq:wilson_loop_upper_bound}.

\end{widetext}

\end{document}